\documentclass[letterpaper,11pt]{article}
\usepackage{url}
\usepackage[nottoc,notlot,notlof]{tocbibind}
\usepackage{tabularx} 
\usepackage{amsmath,amsthm,amssymb}  
\usepackage{amsmath}
\usepackage{amsfonts}
\usepackage{amssymb}
\usepackage{graphicx} 
\usepackage[margin=0.8in,letterpaper]{geometry} 
\usepackage{cite} 
\usepackage[final]{hyperref} 
\usepackage[english]{babel}
\usepackage{caption}
\usepackage{float}
\usepackage[table]{xcolor}
\usepackage{array}
\usepackage{algorithm}
\usepackage{algpseudocode}
\usepackage{cases}
\usepackage{titlesec}
\usepackage{tikz}
\usepackage[utf8]{inputenc}
\usepackage[T1]{fontenc}

\usepackage{mathtools}
\usepackage{extarrows}
\usepackage{booktabs} 
\usepackage{dsfont}

\newtheorem{theorem}{Theorem}
\newtheorem{lemma}{Lemma}
\newtheorem{proposition}{Proposition}
\newtheorem{corollary}{Corollary}

\newtheorem{example}{Example}

\usepackage{microtype}      
\usepackage{lipsum}
\usepackage{nicefrac}
\usepackage{fancyhdr}       

\author{
 El-Mehdi Mehiri \\
  USTHB, Faculty of Mathematics\\
  RECITS Laboratory\\
  BP 32, El Alia 16111, Bab Ezzouar\\
Algiers, Algeria \\
  \texttt{emehiri@usthb.dz} \\
  \texttt{mehiri314@gmail.com} \\
   \and
 Hacène Belbachir \\
  USTHB, Faculty of Mathematics\\
  RECITS Laboratory\\
  BP 32, El Alia 16111, Bab Ezzouar\\
Algiers, Algeria \\
  \texttt{hbelbachir@usthb.dz} \\
  \texttt{hacenebelbachir@gmail.com } \\
}
\title{\textbf{The weighted Tower of Hanoi}}

\date{}

\begin{document}
	
	\maketitle
 
	\begin{abstract}
	    The weighted Tower of Hanoi is a new generalization of the classical Tower of Hanoi problem, where a move of a disc  between two pegs $i$ and $j$ is  weighted by a positive real $w_{ij}\geq 0$. This new problem generalizes  the concept of finding the minimum number of moves to solve the Tower of Hanoi, to find a sequence of moves with the minimum total cost. We present an optimal dynamic algorithm to solve the weighted Tower of Hanoi problem, we also establish some properties of this problem, as well as its relation with the Tower of Hanoi variants that are based on move restriction. 
	     
	\end{abstract}
    \section{Introduction}
The Tower of Hanoi puzzle was introduced to the world by Edward Lucas in 1884 \cite{lucas1883recreations}, since then those interested in the puzzle have created new variants in order to raise its difficulty level and to study it more, sometimes by increasing the number of pegs \cite{dudeney2002canterbury,frame1941solution,stockmeyer1994variations} or considering an arbitrary initial and final states \cite{hinz2018tower} (which can be considered as a generalization) and at other times, by forbidding certain movements of discs between certain pegs \cite{stockmeyer1994variations,atkinson1981cyclic,sapir2004tower} (which is a kind of restriction of the problem), while others have chosen to allow the discs to be placed on top of smaller  discs \cite{wood1980towers} (which is considered a relaxation of the problem). For more on the Tower of Hanoi problem, we point the interested reader to \cite{hinz2018tower}.

The original Lucas's Tower of Hanoi  problem is stated as follows: 
Consider three pegs, source peg, intermediate peg and destination peg; and $n\in \mathbb{N}_{0}$ discs of distinct diameters. Initially, all discs are stacked on the first peg (the source) ordered by their diameters, with the smallest one on top and the largest one on the bottom. The goal is to transfer the $n$ discs to the third peg (the destination) using the minimum number of moves, and respecting the following rules:

\begin{itemize}
\item[$(i)$] at each step only one disc can be moved;
       \item[$(ii)$] the disc moved must be a topmost disc;
       \item[$(iii)$] a disc cannot reside on a smaller one.
\end{itemize}

Solving the problem requires exactly $2^{n}-1$ moves, which can be shown to be optimal \cite{hinz2018tower}. However, there has been continued interest in the problem from several viewpoints and not just the minimum number of moves.

The Tower of Hanoi has its applications in many fields such as didactic of mathematics,  psychology, industry and logistics. The Tower of Hanoi and its variants are used to introduce the concept of mathematical induction to students in computer science, and discrete mathematics. \\

The Tower of Hanoi is a model commonly used in cognitive psychology and neuropsychology to study and examine problem-solving skills which can be tested by calculating moves and strategies while predicting possible outcomes,  we reefer to \cite{hinz2009mathematical,anderson2005tracing,fansher2022effect} for more information about the Tower of Hanoi applications in psychology. The Tower of Hanoi puzzle can be used to model a class of logistic problems which called Pile problems  \cite{hempel2006pile}, in \cite{aguilar2016application} authors discussed the application of Towers of Hanoi in logistics management and in Pile problems in particular. However, the Tower of Hanoi can be seen as a scheduling problem, where the hand playing with discs is a machine or a crane in a big harbor and discs are the containers in the harbor. Only three zones in the harbor are used to stack containers which represents the three pegs, the crane move the containers from a zone to another while respecting the rules $(i-iii)$ of the Tower of Hanoi, where the discs in this case are the containers of different sizes. Given an initial state of the containers in harbor, the crane is asked to move containers to reach another state while minimizing the number of moves of containers between the three zones, this problem is an optimization problem.

In order to expand research on the Tower of Hanoi problem  we present in this paper a new optimization problem which is a new generalization of the Tower of Hanoi problem that has never been considered before in the literature, in which the movements of the discs between the pegs are weighted using a definite weighting function of the set of pegs in the set of positive reals
$\mathbb{R}_{+}$, which makes the problem less trivial than the classical version of the problem. We call this new problem by the weighted Tower of Hanoi, this optimization problem models the harbor problem mentioned before more effectively, because it is more realistic to assume that a movement of a container from a zone to another have a cost, so the goal here is to minimize the total cost and not the number of moves.

\section{Description of the problem}
Consider the same elements as the original Tower of Hanoi problem, three pegs $i$, $j$ and $k$, and a set of $n\in \mathbb{N}_{0}$ discs of different diameters, and the same $(i)-(iii)$ rules of the classical version. In this new problem, a move  from peg $i$ to peg $j$ is weighted by a cost $w_{ij}\in\mathbb{R}^{+}$, where the objective is to transfer the tower of the $n$ discs from the source peg to the destination peg by minimizing the total cost (sum of costs) of the moves used in the solution.\\
The following digraph is called the weighted movement digraph, where the vertices represent  the pegs and the weighted arcs are the movements between pegs with their costs.

\begin{figure}[H]
        \centering
    
\tikzset{every picture/.style={line width=0.75pt}} 

\begin{tikzpicture}[x=0.75pt,y=0.75pt,yscale=-1,xscale=1]

\draw  [fill={rgb, 255:red, 0; green, 0; blue, 0 }  ,fill opacity=1 ] (286,56.5) .. controls (286,52.36) and (289.36,49) .. (293.5,49) .. controls (297.64,49) and (301,52.36) .. (301,56.5) .. controls (301,60.64) and (297.64,64) .. (293.5,64) .. controls (289.36,64) and (286,60.64) .. (286,56.5) -- cycle ;
\draw  [fill={rgb, 255:red, 0; green, 0; blue, 0 }  ,fill opacity=1 ] (408,257.5) .. controls (408,253.36) and (411.36,250) .. (415.5,250) .. controls (419.64,250) and (423,253.36) .. (423,257.5) .. controls (423,261.64) and (419.64,265) .. (415.5,265) .. controls (411.36,265) and (408,261.64) .. (408,257.5) -- cycle ;
\draw  [fill={rgb, 255:red, 0; green, 0; blue, 0 }  ,fill opacity=1 ] (169,255.5) .. controls (169,251.36) and (172.36,248) .. (176.5,248) .. controls (180.64,248) and (184,251.36) .. (184,255.5) .. controls (184,259.64) and (180.64,263) .. (176.5,263) .. controls (172.36,263) and (169,259.64) .. (169,255.5) -- cycle ;

\draw  [color={rgb, 255:red, 0; green, 0; blue, 0 }  ,draw opacity=0 ][fill={rgb, 255:red, 0; green, 0; blue, 0 }  ,fill opacity=0 ]  (295, 56) circle [x radius= 13.6, y radius= 13.6]   ;
\draw (289,48.4) node [anchor=north west][inner sep=0.75pt]    {};
\draw  [color={rgb, 255:red, 0; green, 0; blue, 0 }  ,draw opacity=0 ]  (416, 256) circle [x radius= 13.6, y radius= 13.6]   ;
\draw (410,248.4) node [anchor=north west][inner sep=0.75pt]    {};
\draw  [color={rgb, 255:red, 0; green, 0; blue, 0 }  ,draw opacity=0 ]  (175, 256) circle [x radius= 13.6, y radius= 13.6]   ;
\draw (169,248.4) node [anchor=north west][inner sep=0.75pt]    {};
\draw (240,156.4) node [anchor=north west][inner sep=0.75pt]    {$w_{13}$};
\draw (211,139.4) node [anchor=north west][inner sep=0.75pt]    {$w_{31}$};
\draw (364,140.4) node [anchor=north west][inner sep=0.75pt]    {$w_{12}$};
\draw (324,156.4) node [anchor=north west][inner sep=0.75pt]    {$w_{21}$};
\draw (283,261.4) node [anchor=north west][inner sep=0.75pt]    {$w_{23}$};
\draw (282,232.4) node [anchor=north west][inner sep=0.75pt]    {$w_{32}$};
\draw (307,39.9) node [anchor=north west][inner sep=0.75pt]    {$1$};
\draw (151,238.9) node [anchor=north west][inner sep=0.75pt]    {$3$};
\draw (430,240.9) node [anchor=north west][inner sep=0.75pt]    {$2$};
\draw    (175.39,242.4) .. controls (187.63,212.87) and (201.69,184.44) .. (217.57,157.09)(227.76,140.11) .. controls (244.23,113.52) and (262.48,88.01) .. (282.51,63.56) ;
\draw [shift={(283.18,62.74)}, rotate = 129.42] [color={rgb, 255:red, 0; green, 0; blue, 0 }  ][line width=0.75]    (10.93,-3.29) .. controls (6.95,-1.4) and (3.31,-0.3) .. (0,0) .. controls (3.31,0.3) and (6.95,1.4) .. (10.93,3.29)   ;
\draw    (294.61,69.6) .. controls (282.37,99.13) and (268.31,127.56) .. (252.43,154.91)(242.24,171.89) .. controls (225.77,198.48) and (207.52,223.99) .. (187.49,248.44) ;
\draw [shift={(186.82,249.26)}, rotate = 309.42] [color={rgb, 255:red, 0; green, 0; blue, 0 }  ][line width=0.75]    (10.93,-3.29) .. controls (6.95,-1.4) and (3.31,-0.3) .. (0,0) .. controls (3.31,0.3) and (6.95,1.4) .. (10.93,3.29)   ;
\draw    (306.8,62.77) .. controls (327.05,87.55) and (345.55,113.39) .. (362.32,140.28)(372.59,157.26) .. controls (388.45,184.29) and (402.61,212.34) .. (415.07,241.43) ;
\draw [shift={(415.48,242.41)}, rotate = 246.9] [color={rgb, 255:red, 0; green, 0; blue, 0 }  ][line width=0.75]    (10.93,-3.29) .. controls (6.95,-1.4) and (3.31,-0.3) .. (0,0) .. controls (3.31,0.3) and (6.95,1.4) .. (10.93,3.29)   ;
\draw    (404.04,262.48) .. controls (372.72,266.58) and (341.4,268.9) .. (310.08,269.45)(280.92,269.45) .. controls (249.96,268.91) and (219.01,266.63) .. (188.05,262.62) ;
\draw [shift={(186.96,262.48)}, rotate = 7.46] [color={rgb, 255:red, 0; green, 0; blue, 0 }  ][line width=0.75]    (10.93,-3.29) .. controls (6.95,-1.4) and (3.31,-0.3) .. (0,0) .. controls (3.31,0.3) and (6.95,1.4) .. (10.93,3.29)   ;
\draw    (404.2,249.23) .. controls (383.95,224.45) and (365.45,198.61) .. (348.68,171.72)(338.41,154.74) .. controls (322.55,127.71) and (308.39,99.66) .. (295.93,70.57) ;
\draw [shift={(295.52,69.59)}, rotate = 66.9] [color={rgb, 255:red, 0; green, 0; blue, 0 }  ][line width=0.75]    (10.93,-3.29) .. controls (6.95,-1.4) and (3.31,-0.3) .. (0,0) .. controls (3.31,0.3) and (6.95,1.4) .. (10.93,3.29)   ;
\draw    (186.96,249.52) .. controls (218.28,245.42) and (249.6,243.1) .. (280.92,242.55)(310.08,242.55) .. controls (341.04,243.09) and (371.99,245.37) .. (402.95,249.38) ;
\draw [shift={(404.04,249.52)}, rotate = 187.46] [color={rgb, 255:red, 0; green, 0; blue, 0 }  ][line width=0.75]    (10.93,-3.29) .. controls (6.95,-1.4) and (3.31,-0.3) .. (0,0) .. controls (3.31,0.3) and (6.95,1.4) .. (10.93,3.29)   ;

\end{tikzpicture}

        \caption{The weighted movement digraph}
        \label{fig:graph_complet_avec_poids}
    \end{figure}
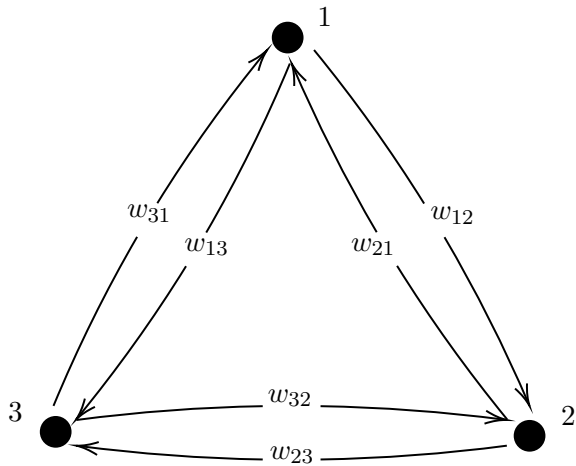
It is easy to see that the optimal solution is not necessarily unique, for example, if we take the costs as follows $w_{13}=w_{12}+w_{23}$ then the problem of transferring a tower of a single disc $n=1$ located on peg $1$ towards peg $3$, have two optimal solutions with the same cost $w_{13}=w_{12}+w_{23}$, but they differ by the number of moves, one solution solves the problem using one move, while the other solves it using two moves.\\

This last remark gives rise to a new question, what is the solution that costs the minimum and solves the problem using the minimum number of moves? In this case, the problem becomes a bi-objective optimization problem, in which the two objectives are the minimization of cost and moves at the same time.\\
In the next section we present an optimal algorithm that solves the problem by giving the priority to the minimum total cost where it searches  among solutions with the minimum total cost the one that have the minimum number of moves.
\section{A recursive formula for the optimal total cost and an optimal algorithm}
Let $C_{n}^{i,j}$ be the total cost of the optimal solution(s) of a weighted Tower of Hanoi with $n$ discs where $i$ is the source peg and $j$ the destination peg with $i ,j\in\{1,2,3\}$, $i\neq j$ and $k=6-i-j$ is the intermediate peg.\\
The total cost $C_{n}^{i,j}$ of the optimal solution(s) can be calculated recursively in terms of the number of discs $n$, the recursion formula is given as follows.
\begin{theorem}\label{theo1}
For all $n\in\mathbb{N}$, and $i$, $j$ and $k$ are respectively the source, destination and intermediate pegs, 
\begin{equation}\label{1}C_{n}^{i,j}=
\begin{cases}
\min\{w_{ij},w_{ik}+w_{kj}\}&if\; n=1,\\
\min\{C_{n-1}^{i,k}+C_{n-1}^{k,j}+w_{ij},2C_{n-1}^{i,j}+C_{ n-1}^{j,i}+w_{ik}+w_{kj}\}&otherwise.
\end{cases}
\end{equation}
\end{theorem}
Note that for the trivial case $n=0$ we have $C_{0}^{i,j}=0$.
\begin{proof}
Using induction on the number of discs $n\in \mathbb{N}$, we present the following proof.\\
Consider $n$ discs, the source and destination pegs are $i$ and $j$ respectively, which implies that $k=6-i-j$ is the intermediate peg, the objective is to find a recursive formula for $C_{n}^{i,j}$.\\
For $n=1$, we have two candidate solutions, which are either to move the unique disc from peg $i$ directly to peg $j$ with a cost $w_{ij}$ or to move it first to the intermediate peg $k$ then to peg $j$ with a total costs equal to $w_{ik}+w_{kj}$, then we have $C_{1}^{i,j}=\min\{w_{ij},w_{ik}+w_{kj}\}$.\\
Induction hypothesis, suppose that the optimal solutions of any weighted Tower of Hanoi having at most $k<n$ discs, have a cost equal to $C_{k}^{i,j}$, then  show that the weighted Tower of Hanoi with $n$ discs have optimal solutions with total cost equals to $C_{n}^{i,j}$.\\
To solve this problem with more than one disc, we have to find a way to move the biggest disc that have a diameter equal to $n$. To do this, we have two ways, the first way is to move the biggest disc directly from source peg $i$ to the destination peg $j$ by transferring the  $(n-1)$ smaller discs to the intermediate peg $k$, then when the biggest disc is on the destination peg, the smaller $(n-1)$ discs can be transferred from intermediate peg $k$ to destination peg $j$, this solution costs $C_{n-1}^{i,k}+C_{n-1}^{k,j}+w_{ij}$ according to the induction hypothesis. The second way is  to move the biggest disc from source peg $i$ to intermediate peg $k$ and then to  destination peg $j$ by transferring the $(n-1)$ smaller discs from source peg $i$ to  destination peg $j$ which allows the biggest peg to be transferred from source  
peg $i$  to  intermediate peg $k$, and then transferring them  (the $n-1$ smaller  discs) from destination peg $j$ to source peg $i$, which liberate the destination peg $j$ so that the biggest disc can be moved there, and finally moving the $(n-1)$ smaller discs from source peg $i$ to  destination peg $j$, this solution costs according to induction hypothesis $2C_{n-1}^{i,j}+C_{ n-1}^{j,i}+w_{ik}+w_{kj}$. The choice between these two solutions depends on which one costs less, hence the optimal solution(s) have a cost equal to $C_{n}^{i,j}=\min\{C_{n-1}^{i,k}+C_{n-1}^{k,j}+w_{ij},2C_{n-1}^{i,j}+C_{ n-1}^{j,i}+w_{ik}+w_{kj}\}$.
\end{proof}

The proof of Theorem \ref{theo1}, can be used to construct a recursive optimal algorithm to solve the weighted Tower of Hanoi problem efficiently, the algorithm is presented bellow.
\begin{algorithm}[H]
\caption{The Weighted Tower of Hanoi dynamic programming algorithm}
\label{WTH2}
\begin{algorithmic}[1]
\Procedure{WTHD}{$n,i,j$}
\State $k:=6-i-j;$
\If{$n=1$}
\If{$w_{ij}\leq w_{ik}+w_{kj}$}
\State move disc $d_{n}$ from peg $i$ to peg $j$;
\Else
\State move disc $d_{n}$ from peg $i$ to peg $k$;
\State move disc $d_{n}$ from peg $k$ to peg $j$;
\EndIf
\Else
\If{$C_{n-1}^{i,k}+C_{n-1}^{k,j}+w_{ij}\leq 2C_{n-1}^{i,j}+C_ {n-1}^{j,i}+w_{ik}+w_{kj}$}
\State \textsc{WTHD}($n-1,i,k$);
\State move disc $d_{n}$ from peg $i$ to peg $j$;
\State \textsc{WTHD}($n-1,k,j$);
\Else
\State \textsc{WTHD}($n-1,i,j$);
\State move disc $d_{n}$ from peg $i$ to peg $k$;
\State \textsc{WTHD}($n-1,j,i$);
\State move disc $d_{n}$ from peg $k$ to peg $j$;
\State \textsc{WTHD}($n-1,i,j$);
\EndIf
\EndIf
\EndProcedure
\end{algorithmic}
\end{algorithm}

\begin{corollary}
The optimal solution provided by the \textsc{WTHD} algorithm has the minimum number of moves among the optimal solutions.
\end{corollary}
\begin{proof}
Whenever having the situation $w_{ij}=w_{ik}+w_{kj}$, it is clear that choosing to move the disc from peg $i$ to peg $j$ directly will generate fewer  moves then choosing to move it from peg $i$ to peg $k$ then to peg $j$. The same thing in the situation where $C_{n-1}^{i, k}+C_{n-1}^{k,j}+w_{ij}= 2C_{n-1}^{i,j}+C_{n-1}^{j,i}+w_{ik }+w_{kj}$, choosing to move discs according to the left side of the equality will generate less number of moves, because the left side corresponds to moving the biggest disc directly from source peg to destination peg using only one move rather than moving it through intermediate peg as corresponds to the right side of the equality where two moves are used to move the biggest disc to destination peg. Whatever the situation, algorithm \textsc{WTHD} always takes the choice with the least number of moves and this is due to the use of the comparison operator $\leq$ when comparing between $w_{ij}$ and $w_{ik}+w_{kj}$, and between $C_{n-1}^{i, k}+C_{n-1}^{k,j}+w_{ij}$ and $2C_{n-1}^{i,j}+C_{n-1}^{j,i}+w_{ik }+w_{kj}$.
\end{proof}

Let $u_{n}^{i,j}$ be the number of moves of an optimal solution of the WTH with $n$ discs, then we have the following result.
\begin{proposition} For all $n\geq 0$
$$2^{n}-1\leq u_{n}^{i,j} \leq 3^{n}-1 .$$
\end{proposition}
\begin{proof}
It is known that the shortest and the longest solution in terms of number of moves are the solutions corresponding to  the classical Tower of Hanoi  and linear variant of the Tower of Hanoi respectively \cite{hinz2018tower}, so it suffices to notice that any solution of this problem must have a number of moves between the number of moves of these two solutions, which are $2^{n}-1$ and $3^{n}-1$ respectively.
\end{proof}

Let $v_{n}$ denote the number of sub-problems needed to be solved by algorithm \textsc{WTHD} to solve the weighted Tower of Hanoi of $n$ discs with two arbitrary source and destination pegs. According to relation \ref{theo1}, to calculate $C_{n}^{i,j}$ we need to calculate four sub-problems $C_{n-1}^{i,j}$, $C_{n-1}^{i,k}$, $C_{n-1}^{j,i}$ and $C_{n-1}^{k,j}$. Figure \ref{fig:tree} shows the structure of search tree when applying relation \ref{theo1}.
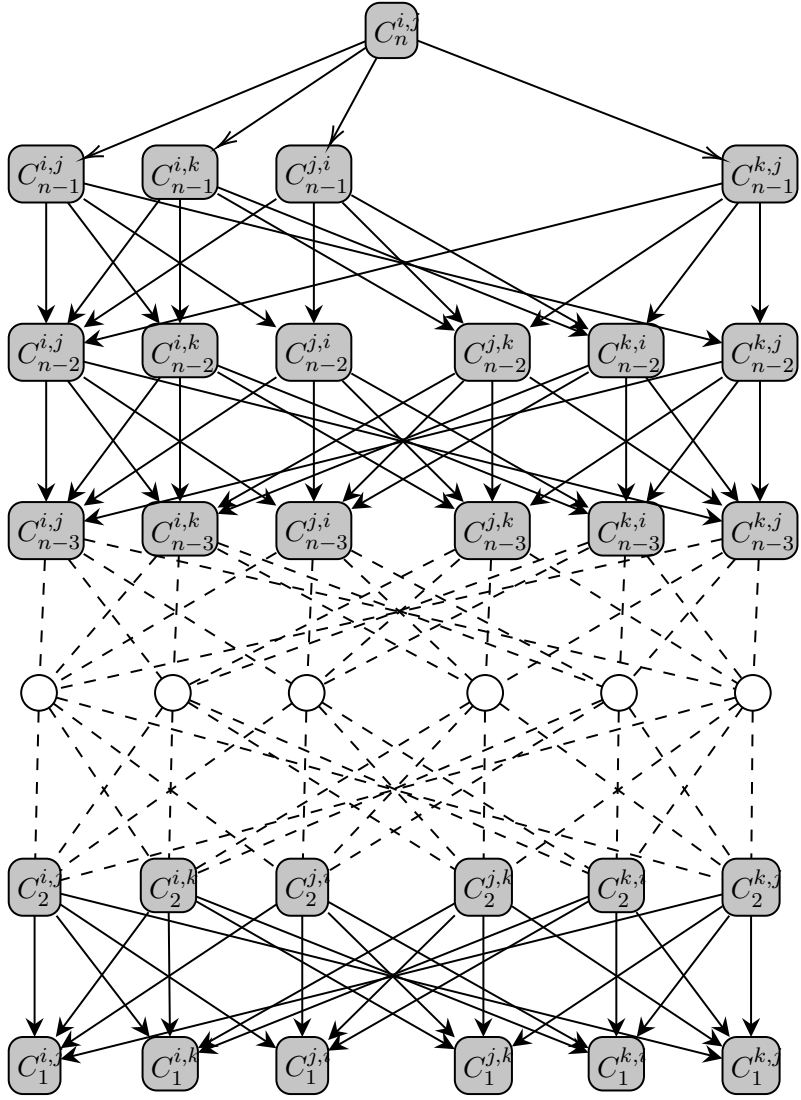
\begin{figure}
    \centering
\tikzset{every picture/.style={line width=0.75pt}} 

\begin{tikzpicture}[scale=0.9,x=0.75pt,y=0.75pt,yscale=-1,xscale=1]

\draw  [fill={rgb, 255:red, 197; green, 197; blue, 197 }  ,fill opacity=1 ]  (299,7) .. controls (299,2.58) and (302.58,-1) .. (307,-1) -- (320,-1) .. controls (324.42,-1) and (328,2.58) .. (328,7) -- (328,22) .. controls (328,26.42) and (324.42,30) .. (320,30) -- (307,30) .. controls (302.58,30) and (299,26.42) .. (299,22) -- cycle  ;
\draw (302,3.4) node [anchor=north west][inner sep=0.75pt]    {$C_{n}^{i,j}$};
\draw  [fill={rgb, 255:red, 197; green, 197; blue, 197 }  ,fill opacity=1 ]  (99,87) .. controls (99,82.58) and (102.58,79) .. (107,79) -- (133,79) .. controls (137.42,79) and (141,82.58) .. (141,87) -- (141,103) .. controls (141,107.42) and (137.42,111) .. (133,111) -- (107,111) .. controls (102.58,111) and (99,107.42) .. (99,103) -- cycle  ;
\draw (102,83.4) node [anchor=north west][inner sep=0.75pt]    {$C_{n-1}^{i,j}$};
\draw  [fill={rgb, 255:red, 197; green, 197; blue, 197 }  ,fill opacity=1 ]  (249,87) .. controls (249,82.58) and (252.58,79) .. (257,79) -- (283,79) .. controls (287.42,79) and (291,82.58) .. (291,87) -- (291,103) .. controls (291,107.42) and (287.42,111) .. (283,111) -- (257,111) .. controls (252.58,111) and (249,107.42) .. (249,103) -- cycle  ;
\draw (252,83.4) node [anchor=north west][inner sep=0.75pt]    {$C_{n-1}^{j,i}$};
\draw  [fill={rgb, 255:red, 197; green, 197; blue, 197 }  ,fill opacity=1 ]  (174,87) .. controls (174,82.58) and (177.58,79) .. (182,79) -- (208,79) .. controls (212.42,79) and (216,82.58) .. (216,87) -- (216,101) .. controls (216,105.42) and (212.42,109) .. (208,109) -- (182,109) .. controls (177.58,109) and (174,105.42) .. (174,101) -- cycle  ;
\draw (177,83.4) node [anchor=north west][inner sep=0.75pt]    {$C_{n-1}^{i,k}$};
\draw  [fill={rgb, 255:red, 197; green, 197; blue, 197 }  ,fill opacity=1 ]  (499,87) .. controls (499,82.58) and (502.58,79) .. (507,79) -- (533,79) .. controls (537.42,79) and (541,82.58) .. (541,87) -- (541,103) .. controls (541,107.42) and (537.42,111) .. (533,111) -- (507,111) .. controls (502.58,111) and (499,107.42) .. (499,103) -- cycle  ;
\draw (502,83.4) node [anchor=north west][inner sep=0.75pt]    {$C_{n-1}^{k,j}$};
\draw  [fill={rgb, 255:red, 197; green, 197; blue, 197 }  ,fill opacity=1 ]  (99,187) .. controls (99,182.58) and (102.58,179) .. (107,179) -- (133,179) .. controls (137.42,179) and (141,182.58) .. (141,187) -- (141,203) .. controls (141,207.42) and (137.42,211) .. (133,211) -- (107,211) .. controls (102.58,211) and (99,207.42) .. (99,203) -- cycle  ;
\draw (102,183.4) node [anchor=north west][inner sep=0.75pt]    {$C_{n-2}^{i,j}$};
\draw  [fill={rgb, 255:red, 197; green, 197; blue, 197 }  ,fill opacity=1 ]  (174,187) .. controls (174,182.58) and (177.58,179) .. (182,179) -- (208,179) .. controls (212.42,179) and (216,182.58) .. (216,187) -- (216,201) .. controls (216,205.42) and (212.42,209) .. (208,209) -- (182,209) .. controls (177.58,209) and (174,205.42) .. (174,201) -- cycle  ;
\draw (177,183.4) node [anchor=north west][inner sep=0.75pt]    {$C_{n-2}^{i,k}$};
\draw  [fill={rgb, 255:red, 197; green, 197; blue, 197 }  ,fill opacity=1 ]  (249,187) .. controls (249,182.58) and (252.58,179) .. (257,179) -- (283,179) .. controls (287.42,179) and (291,182.58) .. (291,187) -- (291,203) .. controls (291,207.42) and (287.42,211) .. (283,211) -- (257,211) .. controls (252.58,211) and (249,207.42) .. (249,203) -- cycle  ;
\draw (252,183.4) node [anchor=north west][inner sep=0.75pt]    {$C_{n-2}^{j,i}$};
\draw  [fill={rgb, 255:red, 197; green, 197; blue, 197 }  ,fill opacity=1 ]  (349,187) .. controls (349,182.58) and (352.58,179) .. (357,179) -- (383,179) .. controls (387.42,179) and (391,182.58) .. (391,187) -- (391,203) .. controls (391,207.42) and (387.42,211) .. (383,211) -- (357,211) .. controls (352.58,211) and (349,207.42) .. (349,203) -- cycle  ;
\draw (352,183.4) node [anchor=north west][inner sep=0.75pt]    {$C_{n-2}^{j,k}$};
\draw  [fill={rgb, 255:red, 197; green, 197; blue, 197 }  ,fill opacity=1 ]  (424,187) .. controls (424,182.58) and (427.58,179) .. (432,179) -- (458,179) .. controls (462.42,179) and (466,182.58) .. (466,187) -- (466,201) .. controls (466,205.42) and (462.42,209) .. (458,209) -- (432,209) .. controls (427.58,209) and (424,205.42) .. (424,201) -- cycle  ;
\draw (427,183.4) node [anchor=north west][inner sep=0.75pt]    {$C_{n-2}^{k,i}$};
\draw  [fill={rgb, 255:red, 197; green, 197; blue, 197 }  ,fill opacity=1 ]  (499,187) .. controls (499,182.58) and (502.58,179) .. (507,179) -- (533,179) .. controls (537.42,179) and (541,182.58) .. (541,187) -- (541,203) .. controls (541,207.42) and (537.42,211) .. (533,211) -- (507,211) .. controls (502.58,211) and (499,207.42) .. (499,203) -- cycle  ;
\draw (502,183.4) node [anchor=north west][inner sep=0.75pt]    {$C_{n-2}^{k,j}$};
\draw  [fill={rgb, 255:red, 197; green, 197; blue, 197 }  ,fill opacity=1 ]  (99,287) .. controls (99,282.58) and (102.58,279) .. (107,279) -- (133,279) .. controls (137.42,279) and (141,282.58) .. (141,287) -- (141,303) .. controls (141,307.42) and (137.42,311) .. (133,311) -- (107,311) .. controls (102.58,311) and (99,307.42) .. (99,303) -- cycle  ;
\draw (102,283.4) node [anchor=north west][inner sep=0.75pt]    {$C_{n-3}^{i,j}$};
\draw  [fill={rgb, 255:red, 197; green, 197; blue, 197 }  ,fill opacity=1 ]  (174,287) .. controls (174,282.58) and (177.58,279) .. (182,279) -- (208,279) .. controls (212.42,279) and (216,282.58) .. (216,287) -- (216,301) .. controls (216,305.42) and (212.42,309) .. (208,309) -- (182,309) .. controls (177.58,309) and (174,305.42) .. (174,301) -- cycle  ;
\draw (177,283.4) node [anchor=north west][inner sep=0.75pt]    {$C_{n-3}^{i,k}$};
\draw  [fill={rgb, 255:red, 197; green, 197; blue, 197 }  ,fill opacity=1 ]  (249,287) .. controls (249,282.58) and (252.58,279) .. (257,279) -- (283,279) .. controls (287.42,279) and (291,282.58) .. (291,287) -- (291,303) .. controls (291,307.42) and (287.42,311) .. (283,311) -- (257,311) .. controls (252.58,311) and (249,307.42) .. (249,303) -- cycle  ;
\draw (252,283.4) node [anchor=north west][inner sep=0.75pt]    {$C_{n-3}^{j,i}$};
\draw  [fill={rgb, 255:red, 197; green, 197; blue, 197 }  ,fill opacity=1 ]  (349,287) .. controls (349,282.58) and (352.58,279) .. (357,279) -- (383,279) .. controls (387.42,279) and (391,282.58) .. (391,287) -- (391,303) .. controls (391,307.42) and (387.42,311) .. (383,311) -- (357,311) .. controls (352.58,311) and (349,307.42) .. (349,303) -- cycle  ;
\draw (352,283.4) node [anchor=north west][inner sep=0.75pt]    {$C_{n-3}^{j,k}$};
\draw  [fill={rgb, 255:red, 197; green, 197; blue, 197 }  ,fill opacity=1 ]  (424,287) .. controls (424,282.58) and (427.58,279) .. (432,279) -- (458,279) .. controls (462.42,279) and (466,282.58) .. (466,287) -- (466,301) .. controls (466,305.42) and (462.42,309) .. (458,309) -- (432,309) .. controls (427.58,309) and (424,305.42) .. (424,301) -- cycle  ;
\draw (427,283.4) node [anchor=north west][inner sep=0.75pt]    {$C_{n-3}^{k,i}$};
\draw  [fill={rgb, 255:red, 197; green, 197; blue, 197 }  ,fill opacity=1 ]  (499,287) .. controls (499,282.58) and (502.58,279) .. (507,279) -- (533,279) .. controls (537.42,279) and (541,282.58) .. (541,287) -- (541,303) .. controls (541,307.42) and (537.42,311) .. (533,311) -- (507,311) .. controls (502.58,311) and (499,307.42) .. (499,303) -- cycle  ;
\draw (502,283.4) node [anchor=north west][inner sep=0.75pt]    {$C_{n-3}^{k,j}$};
\draw  [fill={rgb, 255:red, 197; green, 197; blue, 197 }  ,fill opacity=1 ]  (99,487) .. controls (99,482.58) and (102.58,479) .. (107,479) -- (120,479) .. controls (124.42,479) and (128,482.58) .. (128,487) -- (128,503) .. controls (128,507.42) and (124.42,511) .. (120,511) -- (107,511) .. controls (102.58,511) and (99,507.42) .. (99,503) -- cycle  ;
\draw (102,483.4) node [anchor=north west][inner sep=0.75pt]    {$C_{2}^{i,j}$};
\draw  [fill={rgb, 255:red, 197; green, 197; blue, 197 }  ,fill opacity=1 ]  (173,487) .. controls (173,482.58) and (176.58,479) .. (181,479) -- (196,479) .. controls (200.42,479) and (204,482.58) .. (204,487) -- (204,501) .. controls (204,505.42) and (200.42,509) .. (196,509) -- (181,509) .. controls (176.58,509) and (173,505.42) .. (173,501) -- cycle  ;
\draw (176,483.4) node [anchor=north west][inner sep=0.75pt]    {$C_{2}^{i,k}$};
\draw  [fill={rgb, 255:red, 197; green, 197; blue, 197 }  ,fill opacity=1 ]  (249,487) .. controls (249,482.58) and (252.58,479) .. (257,479) -- (270,479) .. controls (274.42,479) and (278,482.58) .. (278,487) -- (278,503) .. controls (278,507.42) and (274.42,511) .. (270,511) -- (257,511) .. controls (252.58,511) and (249,507.42) .. (249,503) -- cycle  ;
\draw (252,483.4) node [anchor=north west][inner sep=0.75pt]    {$C_{2}^{j,i}$};
\draw  [fill={rgb, 255:red, 197; green, 197; blue, 197 }  ,fill opacity=1 ]  (349,487) .. controls (349,482.58) and (352.58,479) .. (357,479) -- (373,479) .. controls (377.42,479) and (381,482.58) .. (381,487) -- (381,503) .. controls (381,507.42) and (377.42,511) .. (373,511) -- (357,511) .. controls (352.58,511) and (349,507.42) .. (349,503) -- cycle  ;
\draw (352,483.4) node [anchor=north west][inner sep=0.75pt]    {$C_{2}^{j,k}$};
\draw  [fill={rgb, 255:red, 197; green, 197; blue, 197 }  ,fill opacity=1 ]  (424,487) .. controls (424,482.58) and (427.58,479) .. (432,479) -- (447,479) .. controls (451.42,479) and (455,482.58) .. (455,487) -- (455,501) .. controls (455,505.42) and (451.42,509) .. (447,509) -- (432,509) .. controls (427.58,509) and (424,505.42) .. (424,501) -- cycle  ;
\draw (427,483.4) node [anchor=north west][inner sep=0.75pt]    {$C_{2}^{k,i}$};
\draw  [fill={rgb, 255:red, 197; green, 197; blue, 197 }  ,fill opacity=1 ]  (499,487) .. controls (499,482.58) and (502.58,479) .. (507,479) -- (523,479) .. controls (527.42,479) and (531,482.58) .. (531,487) -- (531,503) .. controls (531,507.42) and (527.42,511) .. (523,511) -- (507,511) .. controls (502.58,511) and (499,507.42) .. (499,503) -- cycle  ;
\draw (502,483.4) node [anchor=north west][inner sep=0.75pt]    {$C_{2}^{k,j}$};
\draw  [fill={rgb, 255:red, 197; green, 197; blue, 197 }  ,fill opacity=1 ]  (99,587) .. controls (99,582.58) and (102.58,579) .. (107,579) -- (120,579) .. controls (124.42,579) and (128,582.58) .. (128,587) -- (128,603) .. controls (128,607.42) and (124.42,611) .. (120,611) -- (107,611) .. controls (102.58,611) and (99,607.42) .. (99,603) -- cycle  ;
\draw (102,583.4) node [anchor=north west][inner sep=0.75pt]    {$C_{1}^{i,j}$};
\draw  [fill={rgb, 255:red, 197; green, 197; blue, 197 }  ,fill opacity=1 ]  (174,587) .. controls (174,582.58) and (177.58,579) .. (182,579) -- (197,579) .. controls (201.42,579) and (205,582.58) .. (205,587) -- (205,601) .. controls (205,605.42) and (201.42,609) .. (197,609) -- (182,609) .. controls (177.58,609) and (174,605.42) .. (174,601) -- cycle  ;
\draw (177,583.4) node [anchor=north west][inner sep=0.75pt]    {$C_{1}^{i,k}$};
\draw  [fill={rgb, 255:red, 197; green, 197; blue, 197 }  ,fill opacity=1 ]  (249,587) .. controls (249,582.58) and (252.58,579) .. (257,579) -- (270,579) .. controls (274.42,579) and (278,582.58) .. (278,587) -- (278,603) .. controls (278,607.42) and (274.42,611) .. (270,611) -- (257,611) .. controls (252.58,611) and (249,607.42) .. (249,603) -- cycle  ;
\draw (252,583.4) node [anchor=north west][inner sep=0.75pt]    {$C_{1}^{j,i}$};
\draw  [fill={rgb, 255:red, 197; green, 197; blue, 197 }  ,fill opacity=1 ]  (349,587) .. controls (349,582.58) and (352.58,579) .. (357,579) -- (373,579) .. controls (377.42,579) and (381,582.58) .. (381,587) -- (381,603) .. controls (381,607.42) and (377.42,611) .. (373,611) -- (357,611) .. controls (352.58,611) and (349,607.42) .. (349,603) -- cycle  ;
\draw (352,583.4) node [anchor=north west][inner sep=0.75pt]    {$C_{1}^{j,k}$};
\draw  [fill={rgb, 255:red, 197; green, 197; blue, 197 }  ,fill opacity=1 ]  (424,587) .. controls (424,582.58) and (427.58,579) .. (432,579) -- (447,579) .. controls (451.42,579) and (455,582.58) .. (455,587) -- (455,601) .. controls (455,605.42) and (451.42,609) .. (447,609) -- (432,609) .. controls (427.58,609) and (424,605.42) .. (424,601) -- cycle  ;
\draw (427,583.4) node [anchor=north west][inner sep=0.75pt]    {$C_{1}^{k,i}$};
\draw  [fill={rgb, 255:red, 197; green, 197; blue, 197 }  ,fill opacity=1 ]  (499,587) .. controls (499,582.58) and (502.58,579) .. (507,579) -- (523,579) .. controls (527.42,579) and (531,582.58) .. (531,587) -- (531,603) .. controls (531,607.42) and (527.42,611) .. (523,611) -- (507,611) .. controls (502.58,611) and (499,607.42) .. (499,603) -- cycle  ;
\draw (502,583.4) node [anchor=north west][inner sep=0.75pt]    {$C_{1}^{k,j}$};
\draw    (116, 386) circle [x radius= 10, y radius= 10]   ;
\draw (110,383.4) node [anchor=north west][inner sep=0.75pt]  [font=\tiny]  {$$};
\draw    (191, 386) circle [x radius= 10, y radius= 10]   ;
\draw (185,383.4) node [anchor=north west][inner sep=0.75pt]  [font=\tiny]  {$$};
\draw    (266, 386) circle [x radius= 10, y radius= 10]   ;
\draw (260,383.4) node [anchor=north west][inner sep=0.75pt]  [font=\tiny]  {$$};
\draw    (366, 386) circle [x radius= 10, y radius= 10]   ;
\draw (360,383.4) node [anchor=north west][inner sep=0.75pt]  [font=\tiny]  {$$};
\draw    (441, 386) circle [x radius= 10, y radius= 10]   ;
\draw (435,383.4) node [anchor=north west][inner sep=0.75pt]  [font=\tiny]  {$$};
\draw    (516, 386) circle [x radius= 10, y radius= 10]   ;
\draw (510,383.4) node [anchor=north west][inner sep=0.75pt]  [font=\tiny]  {$$};
\draw    (299,20.53) -- (142.85,85.5) ;
\draw [shift={(141,86.26)}, rotate = 337.41] [color={rgb, 255:red, 0; green, 0; blue, 0 }  ][line width=0.75]    (10.93,-3.29) .. controls (6.95,-1.4) and (3.31,-0.3) .. (0,0) .. controls (3.31,0.3) and (6.95,1.4) .. (10.93,3.29)   ;
\draw    (305.12,30) -- (279.6,77.24) ;
\draw [shift={(278.65,79)}, rotate = 298.39] [color={rgb, 255:red, 0; green, 0; blue, 0 }  ][line width=0.75]    (10.93,-3.29) .. controls (6.95,-1.4) and (3.31,-0.3) .. (0,0) .. controls (3.31,0.3) and (6.95,1.4) .. (10.93,3.29)   ;
\draw    (299,24.23) -- (217.66,78.8) ;
\draw [shift={(216,79.91)}, rotate = 326.14] [color={rgb, 255:red, 0; green, 0; blue, 0 }  ][line width=0.75]    (10.93,-3.29) .. controls (6.95,-1.4) and (3.31,-0.3) .. (0,0) .. controls (3.31,0.3) and (6.95,1.4) .. (10.93,3.29)   ;
\draw    (328,20.15) -- (497.14,86.09) ;
\draw [shift={(499,86.81)}, rotate = 201.3] [color={rgb, 255:red, 0; green, 0; blue, 0 }  ][line width=0.75]    (10.93,-3.29) .. controls (6.95,-1.4) and (3.31,-0.3) .. (0,0) .. controls (3.31,0.3) and (6.95,1.4) .. (10.93,3.29)   ;
\draw    (120,111) -- (120,176) ;
\draw [shift={(120,179)}, rotate = 270] [fill={rgb, 255:red, 0; green, 0; blue, 0 }  ][line width=0.08]  [draw opacity=0] (10.72,-5.15) -- (0,0) -- (10.72,5.15) -- (7.12,0) -- cycle    ;
\draw    (141,109) -- (246.5,179.34) ;
\draw [shift={(249,181)}, rotate = 213.69] [fill={rgb, 255:red, 0; green, 0; blue, 0 }  ][line width=0.08]  [draw opacity=0] (10.72,-5.15) -- (0,0) -- (10.72,5.15) -- (7.12,0) -- cycle    ;
\draw    (132.12,111) -- (181.82,176.61) ;
\draw [shift={(183.64,179)}, rotate = 232.85] [fill={rgb, 255:red, 0; green, 0; blue, 0 }  ][line width=0.08]  [draw opacity=0] (10.72,-5.15) -- (0,0) -- (10.72,5.15) -- (7.12,0) -- cycle    ;
\draw    (141,100.25) -- (496.09,189.02) ;
\draw [shift={(499,189.75)}, rotate = 194.04] [fill={rgb, 255:red, 0; green, 0; blue, 0 }  ][line width=0.08]  [draw opacity=0] (10.72,-5.15) -- (0,0) -- (10.72,5.15) -- (7.12,0) -- cycle    ;
\draw    (270,111) -- (270,176) ;
\draw [shift={(270,179)}, rotate = 270] [fill={rgb, 255:red, 0; green, 0; blue, 0 }  ][line width=0.08]  [draw opacity=0] (10.72,-5.15) -- (0,0) -- (10.72,5.15) -- (7.12,0) -- cycle    ;
\draw    (249,109) -- (143.5,179.34) ;
\draw [shift={(141,181)}, rotate = 326.31] [fill={rgb, 255:red, 0; green, 0; blue, 0 }  ][line width=0.08]  [draw opacity=0] (10.72,-5.15) -- (0,0) -- (10.72,5.15) -- (7.12,0) -- cycle    ;
\draw    (286,111) -- (351.88,176.88) ;
\draw [shift={(354,179)}, rotate = 225] [fill={rgb, 255:red, 0; green, 0; blue, 0 }  ][line width=0.08]  [draw opacity=0] (10.72,-5.15) -- (0,0) -- (10.72,5.15) -- (7.12,0) -- cycle    ;
\draw    (291,106.88) -- (421.39,180.64) ;
\draw [shift={(424,182.12)}, rotate = 209.5] [fill={rgb, 255:red, 0; green, 0; blue, 0 }  ][line width=0.08]  [draw opacity=0] (10.72,-5.15) -- (0,0) -- (10.72,5.15) -- (7.12,0) -- cycle    ;
\draw    (195,109) -- (195,176) ;
\draw [shift={(195,179)}, rotate = 270] [fill={rgb, 255:red, 0; green, 0; blue, 0 }  ][line width=0.08]  [draw opacity=0] (10.72,-5.15) -- (0,0) -- (10.72,5.15) -- (7.12,0) -- cycle    ;
\draw    (216,102.4) -- (421.21,184.49) ;
\draw [shift={(424,185.6)}, rotate = 201.8] [fill={rgb, 255:red, 0; green, 0; blue, 0 }  ][line width=0.08]  [draw opacity=0] (10.72,-5.15) -- (0,0) -- (10.72,5.15) -- (7.12,0) -- cycle    ;
\draw    (183.86,109) -- (133.67,176.59) ;
\draw [shift={(131.88,179)}, rotate = 306.6] [fill={rgb, 255:red, 0; green, 0; blue, 0 }  ][line width=0.08]  [draw opacity=0] (10.72,-5.15) -- (0,0) -- (10.72,5.15) -- (7.12,0) -- cycle    ;
\draw    (216,106.12) -- (346.4,181.38) ;
\draw [shift={(349,182.88)}, rotate = 209.99] [fill={rgb, 255:red, 0; green, 0; blue, 0 }  ][line width=0.08]  [draw opacity=0] (10.72,-5.15) -- (0,0) -- (10.72,5.15) -- (7.12,0) -- cycle    ;
\draw    (520,111) -- (520,176) ;
\draw [shift={(520,179)}, rotate = 270] [fill={rgb, 255:red, 0; green, 0; blue, 0 }  ][line width=0.08]  [draw opacity=0] (10.72,-5.15) -- (0,0) -- (10.72,5.15) -- (7.12,0) -- cycle    ;
\draw    (499,109) -- (393.5,179.34) ;
\draw [shift={(391,181)}, rotate = 326.31] [fill={rgb, 255:red, 0; green, 0; blue, 0 }  ][line width=0.08]  [draw opacity=0] (10.72,-5.15) -- (0,0) -- (10.72,5.15) -- (7.12,0) -- cycle    ;
\draw    (507.88,111) -- (458.18,176.61) ;
\draw [shift={(456.36,179)}, rotate = 307.15] [fill={rgb, 255:red, 0; green, 0; blue, 0 }  ][line width=0.08]  [draw opacity=0] (10.72,-5.15) -- (0,0) -- (10.72,5.15) -- (7.12,0) -- cycle    ;
\draw    (499,100.25) -- (143.91,189.02) ;
\draw [shift={(141,189.75)}, rotate = 345.96] [fill={rgb, 255:red, 0; green, 0; blue, 0 }  ][line width=0.08]  [draw opacity=0] (10.72,-5.15) -- (0,0) -- (10.72,5.15) -- (7.12,0) -- cycle    ;
\draw    (120,211) -- (120,276) ;
\draw [shift={(120,279)}, rotate = 270] [fill={rgb, 255:red, 0; green, 0; blue, 0 }  ][line width=0.08]  [draw opacity=0] (10.72,-5.15) -- (0,0) -- (10.72,5.15) -- (7.12,0) -- cycle    ;
\draw    (141,209) -- (246.5,279.34) ;
\draw [shift={(249,281)}, rotate = 213.69] [fill={rgb, 255:red, 0; green, 0; blue, 0 }  ][line width=0.08]  [draw opacity=0] (10.72,-5.15) -- (0,0) -- (10.72,5.15) -- (7.12,0) -- cycle    ;
\draw    (132.12,211) -- (181.82,276.61) ;
\draw [shift={(183.64,279)}, rotate = 232.85] [fill={rgb, 255:red, 0; green, 0; blue, 0 }  ][line width=0.08]  [draw opacity=0] (10.72,-5.15) -- (0,0) -- (10.72,5.15) -- (7.12,0) -- cycle    ;
\draw    (141,200.25) -- (496.09,289.02) ;
\draw [shift={(499,289.75)}, rotate = 194.04] [fill={rgb, 255:red, 0; green, 0; blue, 0 }  ][line width=0.08]  [draw opacity=0] (10.72,-5.15) -- (0,0) -- (10.72,5.15) -- (7.12,0) -- cycle    ;
\draw    (195,209) -- (195,276) ;
\draw [shift={(195,279)}, rotate = 270] [fill={rgb, 255:red, 0; green, 0; blue, 0 }  ][line width=0.08]  [draw opacity=0] (10.72,-5.15) -- (0,0) -- (10.72,5.15) -- (7.12,0) -- cycle    ;
\draw    (216,202.4) -- (421.21,284.49) ;
\draw [shift={(424,285.6)}, rotate = 201.8] [fill={rgb, 255:red, 0; green, 0; blue, 0 }  ][line width=0.08]  [draw opacity=0] (10.72,-5.15) -- (0,0) -- (10.72,5.15) -- (7.12,0) -- cycle    ;
\draw    (183.86,209) -- (133.67,276.59) ;
\draw [shift={(131.88,279)}, rotate = 306.6] [fill={rgb, 255:red, 0; green, 0; blue, 0 }  ][line width=0.08]  [draw opacity=0] (10.72,-5.15) -- (0,0) -- (10.72,5.15) -- (7.12,0) -- cycle    ;
\draw    (216,206.12) -- (346.4,281.38) ;
\draw [shift={(349,282.88)}, rotate = 209.99] [fill={rgb, 255:red, 0; green, 0; blue, 0 }  ][line width=0.08]  [draw opacity=0] (10.72,-5.15) -- (0,0) -- (10.72,5.15) -- (7.12,0) -- cycle    ;
\draw    (270,211) -- (270,276) ;
\draw [shift={(270,279)}, rotate = 270] [fill={rgb, 255:red, 0; green, 0; blue, 0 }  ][line width=0.08]  [draw opacity=0] (10.72,-5.15) -- (0,0) -- (10.72,5.15) -- (7.12,0) -- cycle    ;
\draw    (249,209) -- (143.5,279.34) ;
\draw [shift={(141,281)}, rotate = 326.31] [fill={rgb, 255:red, 0; green, 0; blue, 0 }  ][line width=0.08]  [draw opacity=0] (10.72,-5.15) -- (0,0) -- (10.72,5.15) -- (7.12,0) -- cycle    ;
\draw    (286,211) -- (351.88,276.88) ;
\draw [shift={(354,279)}, rotate = 225] [fill={rgb, 255:red, 0; green, 0; blue, 0 }  ][line width=0.08]  [draw opacity=0] (10.72,-5.15) -- (0,0) -- (10.72,5.15) -- (7.12,0) -- cycle    ;
\draw    (291,206.88) -- (421.39,280.64) ;
\draw [shift={(424,282.12)}, rotate = 209.5] [fill={rgb, 255:red, 0; green, 0; blue, 0 }  ][line width=0.08]  [draw opacity=0] (10.72,-5.15) -- (0,0) -- (10.72,5.15) -- (7.12,0) -- cycle    ;
\draw    (370,211) -- (370,276) ;
\draw [shift={(370,279)}, rotate = 270] [fill={rgb, 255:red, 0; green, 0; blue, 0 }  ][line width=0.08]  [draw opacity=0] (10.72,-5.15) -- (0,0) -- (10.72,5.15) -- (7.12,0) -- cycle    ;
\draw    (391,209) -- (496.5,279.34) ;
\draw [shift={(499,281)}, rotate = 213.69] [fill={rgb, 255:red, 0; green, 0; blue, 0 }  ][line width=0.08]  [draw opacity=0] (10.72,-5.15) -- (0,0) -- (10.72,5.15) -- (7.12,0) -- cycle    ;
\draw    (354,211) -- (288.12,276.88) ;
\draw [shift={(286,279)}, rotate = 315] [fill={rgb, 255:red, 0; green, 0; blue, 0 }  ][line width=0.08]  [draw opacity=0] (10.72,-5.15) -- (0,0) -- (10.72,5.15) -- (7.12,0) -- cycle    ;
\draw    (349,206.88) -- (218.61,280.64) ;
\draw [shift={(216,282.12)}, rotate = 330.5] [fill={rgb, 255:red, 0; green, 0; blue, 0 }  ][line width=0.08]  [draw opacity=0] (10.72,-5.15) -- (0,0) -- (10.72,5.15) -- (7.12,0) -- cycle    ;
\draw    (445,209) -- (445,276) ;
\draw [shift={(445,279)}, rotate = 270] [fill={rgb, 255:red, 0; green, 0; blue, 0 }  ][line width=0.08]  [draw opacity=0] (10.72,-5.15) -- (0,0) -- (10.72,5.15) -- (7.12,0) -- cycle    ;
\draw    (424,202.4) -- (218.79,284.49) ;
\draw [shift={(216,285.6)}, rotate = 338.2] [fill={rgb, 255:red, 0; green, 0; blue, 0 }  ][line width=0.08]  [draw opacity=0] (10.72,-5.15) -- (0,0) -- (10.72,5.15) -- (7.12,0) -- cycle    ;
\draw    (456.14,209) -- (506.33,276.59) ;
\draw [shift={(508.12,279)}, rotate = 233.4] [fill={rgb, 255:red, 0; green, 0; blue, 0 }  ][line width=0.08]  [draw opacity=0] (10.72,-5.15) -- (0,0) -- (10.72,5.15) -- (7.12,0) -- cycle    ;
\draw    (424,206.12) -- (293.6,281.38) ;
\draw [shift={(291,282.88)}, rotate = 330.01] [fill={rgb, 255:red, 0; green, 0; blue, 0 }  ][line width=0.08]  [draw opacity=0] (10.72,-5.15) -- (0,0) -- (10.72,5.15) -- (7.12,0) -- cycle    ;
\draw    (520,211) -- (520,276) ;
\draw [shift={(520,279)}, rotate = 270] [fill={rgb, 255:red, 0; green, 0; blue, 0 }  ][line width=0.08]  [draw opacity=0] (10.72,-5.15) -- (0,0) -- (10.72,5.15) -- (7.12,0) -- cycle    ;
\draw    (499,209) -- (393.5,279.34) ;
\draw [shift={(391,281)}, rotate = 326.31] [fill={rgb, 255:red, 0; green, 0; blue, 0 }  ][line width=0.08]  [draw opacity=0] (10.72,-5.15) -- (0,0) -- (10.72,5.15) -- (7.12,0) -- cycle    ;
\draw    (507.88,211) -- (458.18,276.61) ;
\draw [shift={(456.36,279)}, rotate = 307.15] [fill={rgb, 255:red, 0; green, 0; blue, 0 }  ][line width=0.08]  [draw opacity=0] (10.72,-5.15) -- (0,0) -- (10.72,5.15) -- (7.12,0) -- cycle    ;
\draw    (499,200.25) -- (143.91,289.02) ;
\draw [shift={(141,289.75)}, rotate = 345.96] [fill={rgb, 255:red, 0; green, 0; blue, 0 }  ][line width=0.08]  [draw opacity=0] (10.72,-5.15) -- (0,0) -- (10.72,5.15) -- (7.12,0) -- cycle    ;
\draw    (113.5,511) -- (113.5,576) ;
\draw [shift={(113.5,579)}, rotate = 270] [fill={rgb, 255:red, 0; green, 0; blue, 0 }  ][line width=0.08]  [draw opacity=0] (10.72,-5.15) -- (0,0) -- (10.72,5.15) -- (7.12,0) -- cycle    ;
\draw    (128,498.61) -- (496.09,590.29) ;
\draw [shift={(499,591.01)}, rotate = 193.99] [fill={rgb, 255:red, 0; green, 0; blue, 0 }  ][line width=0.08]  [draw opacity=0] (10.72,-5.15) -- (0,0) -- (10.72,5.15) -- (7.12,0) -- cycle    ;
\draw    (125.78,511) -- (176.16,576.62) ;
\draw [shift={(177.98,579)}, rotate = 232.49] [fill={rgb, 255:red, 0; green, 0; blue, 0 }  ][line width=0.08]  [draw opacity=0] (10.72,-5.15) -- (0,0) -- (10.72,5.15) -- (7.12,0) -- cycle    ;
\draw    (128,504.67) -- (246.5,583.67) ;
\draw [shift={(249,585.33)}, rotate = 213.69] [fill={rgb, 255:red, 0; green, 0; blue, 0 }  ][line width=0.08]  [draw opacity=0] (10.72,-5.15) -- (0,0) -- (10.72,5.15) -- (7.12,0) -- cycle    ;
\draw    (188.65,509) -- (189.32,576) ;
\draw [shift={(189.35,579)}, rotate = 269.43] [fill={rgb, 255:red, 0; green, 0; blue, 0 }  ][line width=0.08]  [draw opacity=0] (10.72,-5.15) -- (0,0) -- (10.72,5.15) -- (7.12,0) -- cycle    ;
\draw    (204,500.18) -- (421.21,586.71) ;
\draw [shift={(424,587.82)}, rotate = 201.72] [fill={rgb, 255:red, 0; green, 0; blue, 0 }  ][line width=0.08]  [draw opacity=0] (10.72,-5.15) -- (0,0) -- (10.72,5.15) -- (7.12,0) -- cycle    ;
\draw    (177.36,509) -- (127.17,576.59) ;
\draw [shift={(125.38,579)}, rotate = 306.6] [fill={rgb, 255:red, 0; green, 0; blue, 0 }  ][line width=0.08]  [draw opacity=0] (10.72,-5.15) -- (0,0) -- (10.72,5.15) -- (7.12,0) -- cycle    ;
\draw    (204,502.87) -- (346.4,584.35) ;
\draw [shift={(349,585.84)}, rotate = 209.78] [fill={rgb, 255:red, 0; green, 0; blue, 0 }  ][line width=0.08]  [draw opacity=0] (10.72,-5.15) -- (0,0) -- (10.72,5.15) -- (7.12,0) -- cycle    ;
\draw    (263.5,511) -- (263.5,576) ;
\draw [shift={(263.5,579)}, rotate = 270] [fill={rgb, 255:red, 0; green, 0; blue, 0 }  ][line width=0.08]  [draw opacity=0] (10.72,-5.15) -- (0,0) -- (10.72,5.15) -- (7.12,0) -- cycle    ;
\draw    (249,504.67) -- (130.5,583.67) ;
\draw [shift={(128,585.33)}, rotate = 326.31] [fill={rgb, 255:red, 0; green, 0; blue, 0 }  ][line width=0.08]  [draw opacity=0] (10.72,-5.15) -- (0,0) -- (10.72,5.15) -- (7.12,0) -- cycle    ;
\draw    (278,509.29) -- (346.86,577.13) ;
\draw [shift={(349,579.24)}, rotate = 224.57] [fill={rgb, 255:red, 0; green, 0; blue, 0 }  ][line width=0.08]  [draw opacity=0] (10.72,-5.15) -- (0,0) -- (10.72,5.15) -- (7.12,0) -- cycle    ;
\draw    (278,503.16) -- (421.39,583.81) ;
\draw [shift={(424,585.28)}, rotate = 209.36] [fill={rgb, 255:red, 0; green, 0; blue, 0 }  ][line width=0.08]  [draw opacity=0] (10.72,-5.15) -- (0,0) -- (10.72,5.15) -- (7.12,0) -- cycle    ;
\draw    (365,511) -- (365,576) ;
\draw [shift={(365,579)}, rotate = 270] [fill={rgb, 255:red, 0; green, 0; blue, 0 }  ][line width=0.08]  [draw opacity=0] (10.72,-5.15) -- (0,0) -- (10.72,5.15) -- (7.12,0) -- cycle    ;
\draw    (381,505.67) -- (496.5,582.67) ;
\draw [shift={(499,584.33)}, rotate = 213.69] [fill={rgb, 255:red, 0; green, 0; blue, 0 }  ][line width=0.08]  [draw opacity=0] (10.72,-5.15) -- (0,0) -- (10.72,5.15) -- (7.12,0) -- cycle    ;
\draw    (349,510.76) -- (280.14,578.61) ;
\draw [shift={(278,580.71)}, rotate = 315.43] [fill={rgb, 255:red, 0; green, 0; blue, 0 }  ][line width=0.08]  [draw opacity=0] (10.72,-5.15) -- (0,0) -- (10.72,5.15) -- (7.12,0) -- cycle    ;
\draw    (349,504.03) -- (207.61,583.78) ;
\draw [shift={(205,585.26)}, rotate = 330.57] [fill={rgb, 255:red, 0; green, 0; blue, 0 }  ][line width=0.08]  [draw opacity=0] (10.72,-5.15) -- (0,0) -- (10.72,5.15) -- (7.12,0) -- cycle    ;
\draw    (439.5,509) -- (439.5,576) ;
\draw [shift={(439.5,579)}, rotate = 270] [fill={rgb, 255:red, 0; green, 0; blue, 0 }  ][line width=0.08]  [draw opacity=0] (10.72,-5.15) -- (0,0) -- (10.72,5.15) -- (7.12,0) -- cycle    ;
\draw    (424,500.2) -- (207.79,586.69) ;
\draw [shift={(205,587.8)}, rotate = 338.2] [fill={rgb, 255:red, 0; green, 0; blue, 0 }  ][line width=0.08]  [draw opacity=0] (10.72,-5.15) -- (0,0) -- (10.72,5.15) -- (7.12,0) -- cycle    ;
\draw    (450.71,509) -- (501.24,576.6) ;
\draw [shift={(503.04,579)}, rotate = 233.22] [fill={rgb, 255:red, 0; green, 0; blue, 0 }  ][line width=0.08]  [draw opacity=0] (10.72,-5.15) -- (0,0) -- (10.72,5.15) -- (7.12,0) -- cycle    ;
\draw    (424,502.89) -- (280.6,585.19) ;
\draw [shift={(278,586.68)}, rotate = 330.15] [fill={rgb, 255:red, 0; green, 0; blue, 0 }  ][line width=0.08]  [draw opacity=0] (10.72,-5.15) -- (0,0) -- (10.72,5.15) -- (7.12,0) -- cycle    ;
\draw    (515,511) -- (515,576) ;
\draw [shift={(515,579)}, rotate = 270] [fill={rgb, 255:red, 0; green, 0; blue, 0 }  ][line width=0.08]  [draw opacity=0] (10.72,-5.15) -- (0,0) -- (10.72,5.15) -- (7.12,0) -- cycle    ;
\draw    (499,505.67) -- (383.5,582.67) ;
\draw [shift={(381,584.33)}, rotate = 326.31] [fill={rgb, 255:red, 0; green, 0; blue, 0 }  ][line width=0.08]  [draw opacity=0] (10.72,-5.15) -- (0,0) -- (10.72,5.15) -- (7.12,0) -- cycle    ;
\draw    (502.8,511) -- (452.76,576.61) ;
\draw [shift={(450.94,579)}, rotate = 307.33] [fill={rgb, 255:red, 0; green, 0; blue, 0 }  ][line width=0.08]  [draw opacity=0] (10.72,-5.15) -- (0,0) -- (10.72,5.15) -- (7.12,0) -- cycle    ;
\draw    (499,498.99) -- (130.91,590.66) ;
\draw [shift={(128,591.39)}, rotate = 346.01] [fill={rgb, 255:red, 0; green, 0; blue, 0 }  ][line width=0.08]  [draw opacity=0] (10.72,-5.15) -- (0,0) -- (10.72,5.15) -- (7.12,0) -- cycle    ;
\draw  [dash pattern={on 4.5pt off 4.5pt}]  (119.3,311) -- (116.44,376.01) ;
\draw  [dash pattern={on 4.5pt off 4.5pt}]  (141,308.09) -- (257.51,380.71) ;
\draw  [dash pattern={on 4.5pt off 4.5pt}]  (132.48,311) -- (184.85,378.12) ;
\draw  [dash pattern={on 4.5pt off 4.5pt}]  (141,299.83) -- (506.25,383.76) ;
\draw  [dash pattern={on 4.5pt off 4.5pt}]  (194.35,309) -- (191.43,376.01) ;
\draw  [dash pattern={on 4.5pt off 4.5pt}]  (216,301.85) -- (431.63,382.5) ;
\draw  [dash pattern={on 4.5pt off 4.5pt}]  (182.12,309) -- (122.51,378.41) ;
\draw  [dash pattern={on 4.5pt off 4.5pt}]  (216,305.3) -- (357.19,381.26) ;
\draw  [dash pattern={on 4.5pt off 4.5pt}]  (269.3,311) -- (266.44,376.01) ;
\draw  [dash pattern={on 4.5pt off 4.5pt}]  (249,307.41) -- (124.61,380.91) ;
\draw  [dash pattern={on 4.5pt off 4.5pt}]  (286.88,311) -- (358.74,379.12) ;
\draw  [dash pattern={on 4.5pt off 4.5pt}]  (291,306.18) -- (432.17,381.3) ;
\draw  [dash pattern={on 4.5pt off 4.5pt}]  (369.3,311) -- (366.44,376.01) ;
\draw  [dash pattern={on 4.5pt off 4.5pt}]  (391,308.09) -- (507.51,380.71) ;
\draw  [dash pattern={on 4.5pt off 4.5pt}]  (351.71,311) -- (273.53,379.41) ;
\draw  [dash pattern={on 4.5pt off 4.5pt}]  (349,305.68) -- (199.92,381.47) ;
\draw  [dash pattern={on 4.5pt off 4.5pt}]  (444.35,309) -- (441.43,376.01) ;
\draw  [dash pattern={on 4.5pt off 4.5pt}]  (424,301.61) -- (200.4,382.59) ;
\draw  [dash pattern={on 4.5pt off 4.5pt}]  (456.58,309) -- (509.89,378.08) ;
\draw  [dash pattern={on 4.5pt off 4.5pt}]  (424,304.79) -- (274.9,381.43) ;
\draw  [dash pattern={on 4.5pt off 4.5pt}]  (519.3,311) -- (516.44,376.01) ;
\draw  [dash pattern={on 4.5pt off 4.5pt}]  (499,307.41) -- (374.61,380.91) ;
\draw  [dash pattern={on 4.5pt off 4.5pt}]  (506.11,311) -- (447.56,378.45) ;
\draw  [dash pattern={on 4.5pt off 4.5pt}]  (499,299.73) -- (125.76,383.8) ;
\draw  [dash pattern={on 4.5pt off 4.5pt}]  (115.77,396) -- (113.87,479) ;
\draw  [dash pattern={on 4.5pt off 4.5pt}]  (124.04,391.94) -- (249,484.28) ;
\draw  [dash pattern={on 4.5pt off 4.5pt}]  (121.57,394.3) -- (178.43,479) ;
\draw  [dash pattern={on 4.5pt off 4.5pt}]  (125.65,388.64) -- (499,490.63) ;
\draw  [dash pattern={on 4.5pt off 4.5pt}]  (190.77,396) -- (188.85,479) ;
\draw  [dash pattern={on 4.5pt off 4.5pt}]  (200.17,389.99) -- (424,487.26) ;
\draw  [dash pattern={on 4.5pt off 4.5pt}]  (185.2,394.15) -- (124.88,479) ;
\draw  [dash pattern={on 4.5pt off 4.5pt}]  (199.48,391.31) -- (349,484.98) ;
\draw  [dash pattern={on 4.5pt off 4.5pt}]  (265.77,396) -- (263.87,479) ;
\draw  [dash pattern={on 4.5pt off 4.5pt}]  (257.86,391.82) -- (128,484.64) ;
\draw  [dash pattern={on 4.5pt off 4.5pt}]  (272.72,393.4) -- (350.47,479) ;
\draw  [dash pattern={on 4.5pt off 4.5pt}]  (274.49,391.29) -- (424,484.35) ;
\draw  [dash pattern={on 4.5pt off 4.5pt}]  (365.91,396) -- (365.15,479) ;
\draw  [dash pattern={on 4.5pt off 4.5pt}]  (374.07,391.9) -- (499,483.3) ;
\draw  [dash pattern={on 4.5pt off 4.5pt}]  (359.15,393.28) -- (278,479.58) ;
\draw  [dash pattern={on 4.5pt off 4.5pt}]  (357.46,391.2) -- (204,484.57) ;
\draw  [dash pattern={on 4.5pt off 4.5pt}]  (440.86,396) -- (439.71,479) ;
\draw  [dash pattern={on 4.5pt off 4.5pt}]  (431.8,389.93) -- (204,487.37) ;
\draw  [dash pattern={on 4.5pt off 4.5pt}]  (446.62,394.27) -- (504.14,479) ;
\draw  [dash pattern={on 4.5pt off 4.5pt}]  (432.48,391.23) -- (278,486.1) ;
\draw  [dash pattern={on 4.5pt off 4.5pt}]  (515.91,396) -- (515.15,479) ;
\draw  [dash pattern={on 4.5pt off 4.5pt}]  (507.89,391.85) -- (381,483.45) ;
\draw  [dash pattern={on 4.5pt off 4.5pt}]  (510.22,394.16) -- (450.13,479) ;
\draw  [dash pattern={on 4.5pt off 4.5pt}]  (506.35,388.61) -- (128,491.07) ;

\end{tikzpicture}

    \caption{The structure of the search tree.}
    \label{fig:tree}
\end{figure}

\begin{proposition}\label{propo1}
For all $n\geq 3$ we have 
\begin{equation}
    v_{n}=6^{n-2}+4,
\end{equation}
and $v_{0}=0$, $v_{1}=1$ and $v_{2}=4$.
\end{proposition}
\begin{proof}
It is clear that $v_{0}=0$ and $v_{1}=1$, for case $n=2$,  for any source and destination pegs $i$ and $j$ we have $C_{2}^{i,j}=\min\{C_{1}^{i,k}+C_{1}^{k,j}+w_{ij},2C_{1}^{i,j}+C_{ 1}^{j,i}+w_{ik}+w_{kj}\}$ which means that four sub-problems needed to be evaluated which are $C_{1}^{i,j}$, $C_{1}^{i,k}$, $C_{1}^{j,i}$ and $C_{1}^{k,j}$, hence $v_{2}=4$. For $n\geq 3$ we have 
$$C_{n}^{i,j}=\min\{C_{n-1}^{i,k}+C_{n-1}^{k,j}+w_{ij},2C_{n-1}^{i,j}+C_{ n-1}^{j,i}+w_{ik}+w_{kj}\}, $$
where

\begin{align*}
    C_{n-1}^{i,j}&=\min\{C_{n-2}^{i,k}+C_{n-2}^{k,j}+w_{ij},2C_{n-2}^{i,j}+C_{ n-2}^{j,i}+w_{ik}+w_{kj}\},\\
    C_{n-1}^{i,k}&=\min\{C_{n-2}^{i,j}+C_{n-2}^{j,k}+w_{ik},2C_{n-2}^{i,k}+C_{ n-2}^{k,i}+w_{ij}+w_{jk}\},\\
    C_{n-1}^{j,i}&=\min\{C_{n-2}^{j,k}+C_{n-2}^{k,i}+w_{ji},2C_{n-2}^{j,i}+C_{ n-2}^{i,j}+w_{jk}+w_{ki}\},\\
    C_{n-1}^{k,j}&=\min\{C_{n-2}^{k,i}+C_{n-2}^{i,j}+w_{kj},2C_{n-2}^{k,j}+C_{ n-2}^{j,k}+w_{ki}+w_{ij}\}.\\
\end{align*}
In total we have $4$ sub-problems of order $n-1$,  $C_{n-1}^{i,j}$, $C_{n-1}^{i,k}$, $C_{n-1}^{j,i}$ and $C_{n-1}^{k,j}$ and six sub-problems of order $n-2$ which are $C_{n-2}^{i,j}$, $C_{n-2}^{i,k}$, $C_{n-2}^{j,i}$, $C_{n-2}^{j,k}$, $C_{n-2}^{k,i}$ and $C_{n-2}^{k,j}$. To calculate the $6$ sub-problems of order $h$ we need to evaluate all the $6$ possible sub-problems of order $h-1$ for all $2\leq h\leq n-2$. Therefor we have $6^{n-2}$ sub problems of order $h$ with $1\leq h\leq n-2$ in addition to the $4$ sub-problems of order $n-1$. Hence in total we have $v_{n}=6^{n-2}+4$.
\end{proof}
This last result allows us to find the time complexity of the recursive algorithm \textsc{WTHD}.
\begin{corollary}
The \textsc{WTHD} algorithm solves an instances of the WTH problem on $n$ dics  in a time complexity $\mathcal{O}(6^{n-2})$.
\end{corollary}
\begin{proof}
It is a direct result of Proposition \ref{propo1}.
\end{proof}
Here we list some direct results of Theorem \ref{theo1}, for symmetry reasons we have the follwoing result which is more or less easy to see.
\begin{corollary}

If $w_{ij}=w_{ji}$ for all $i,j\in\{1,2,3\}$, then
\begin{equation}\label{2}C_{n}^{i,j}=C_{n}^{j,i}=
\begin{cases}
\min\{w_{ij},w_{ik}+w_{kj}\}&si\; n=1,\\
\min\{C_{n-1}^{i,k}+C_{n-1}^{k,j}+w_{ij},3C_{n-1}^{i,j}+w_{ ik}+w_{kj}\}&else.
\end{cases}
\end{equation}
\end{corollary}
As we mentioned before in this paper, the classical Tower of Hanoi has a solution with $2^{n}-1$ moves, which is the minimum number of moves needed to transfer a whole tower of $n$ discs from a peg to another according to rules $(i-iii)$. Using this fact, we find the following result.

\begin{corollary}
If $w_{ij}=x\in\mathbb{R}$ for all $i,j\in\{1,2,3\}$, then
$$C_{n}^{i,j}=(2^{n}-1)x.$$
\end{corollary}
\begin{proof}
    When all weights are equal, the problem becomes the minimization of the number of moves, and we know that the minimum number of moves to solve a Tower of Hanoi on $n$ discs  is equal to $2^{n}-1$.
\end{proof}

\begin{example}
Consider a weighted Tower of Hanoi weighted problem, the associated weights are given in the matrix bellow
$$(w_{ij})_{1\leq i,j\leq 3 }=
\begin{pmatrix}
0&3&15\\
8&0&2\\
5&6&0
\end{pmatrix}.
$$
The objective is to transfer a tower of $n=3$ discs from peg $i=1$ to  peg $j=3$ while minimizing the total cost. For all $n\geq 0$ we denote $L_{n}^{i,j}=C_{n-1}^{i,k}+C_{n-1}^{k,j}+w_{ij}$ and $R_{n}^{i,j}=2C_{n-1}^{i,j }+C_{n-1}^{j,i}+w_{ik}+w_{kj}$.
The  dynamic Table \ref{tab} shows the steps that the dynamic recursive relation \ref{1} takes to find the optimal value for $C_{n}^{i,j}$. The values with $*$ are optimal in accordance with the corresponding source and destination pegs while the red color corresponds to the selected sub-problems used to evaluate $C_{3}^{1,3}$.
\begin{table}[H]
\centering
\begin{tabular}{@{}c|cc|cc|cc|cc|cc|cc|@{}}
\cmidrule(l){2-13}
                               &\multicolumn{2}{c|}{1$\longrightarrow$2} &\multicolumn{2}{c|}{1$\longrightarrow$3}  & \multicolumn{2}{c|}{2$\longrightarrow$1} & \multicolumn{2}{c|}{2$\longrightarrow$3} & \multicolumn{2}{c|}{3$\longrightarrow$1} & \multicolumn{2}{c|}{3$\longrightarrow$2} \\ \midrule
\multicolumn{1}{|c|}{n}            & \multicolumn{1}{c|}{$L_{n}^{1,2}$}        & $D_{n}^{1,2}$   & \multicolumn{1}{c|}{$L_{n}^{1,3}$}        & $R_{n}^{1,3}$     & \multicolumn{1}{c|}{$L_{n}^{2,1}$}        & $R_{n}^{2,1}$        & \multicolumn{1}{c|}{$L_{n}^{2,3}$}        & $R_{n}^{2,3}$        & \multicolumn{1}{c|}{$L_{n}^{3,1}$}        & $L_{n}^{3,1}$        & \multicolumn{1}{c|}{$L_{n}^{3,2}$}        & $R_{n}^{3,2}$        \\ \midrule
\multicolumn{1}{|c|}{1}               & \multicolumn{1}{c|}{3*}       & 21 & \multicolumn{1}{c|}{15}       & \textcolor{red}{5*}      & \multicolumn{1}{c|}{8}        & \textcolor{red}{7*}       & \multicolumn{1}{c|}{2*}       & 23       & \multicolumn{1}{c|}{5*}       & 14       & \multicolumn{1}{c|}{\textcolor{red}{6*}}       & 8        \\ \midrule
\multicolumn{1}{|c|}{2}              & \multicolumn{1}{c|}{\textcolor{red}{14*}}      & 34   & \multicolumn{1}{c|}{20*}      & 20*    & \multicolumn{1}{c|}{15*}      & 24       & \multicolumn{1}{c|}{\textcolor{red}{14*}}      & 33       & \multicolumn{1}{c|}{18*}      & 29       & \multicolumn{1}{c|}{14*}      & 22       \\ \midrule
\multicolumn{1}{|c|}{3}              & \multicolumn{1}{c|}{37*}      & 64 & \multicolumn{1}{c|}{\textbf{43*}}      & 63       & \multicolumn{1}{c|}{40*}      & 51       & \multicolumn{1}{c|}{37*}      & 65       & \multicolumn{1}{c|}{34*}      & 70       & \multicolumn{1}{c|}{38*}      & 50       \\ \midrule
\multicolumn{1}{|c|}{$\colon$} & \multicolumn{1}{c|}{$\colon$} & $\colon$ & \multicolumn{1}{c|}{$\colon$} & $\colon$ & \multicolumn{1}{c|}{$\colon$} & $\colon$ & \multicolumn{1}{c|}{$\colon$} & $\colon$ & \multicolumn{1}{c|}{$\colon$} & $\colon$ & \multicolumn{1}{c|}{$\colon$} & $\colon$ \\ 
\end{tabular}
    \caption{The progress of calculating the total cost}

\label{tab}
\end{table}

So the optimal cost to transfer a tower of $n=3$ discs from peg $1$ to peg $3$ 
is $C_{3}^{1,3}=43$ where 
\begin{align*}
    43&=14+14+15=C_{2}^{1,2}+C_{2}^{2 ,3}+w_{1,3}\\
    &=(5+6+3)+(7+5+2)+15 =(C_{1}^{1,3}+C_{1}^{3,2}+w_{1,2})+(C_{1}^{2,1}+C_{1}^{1,3}+w_{2,3})+w_{1,3}.
\end{align*}
\end{example}
\section{The relationship to variants with restricted disc moves}
Variants with restricted disc moves consists of forbidding some  transitions between certain pegs, for example, moves between pegs $1$ and $3$ are forbidden, which implies that the only moves allowed are those between adjacent pegs (if the pegs are aligned from the  peg $1$ to  peg $3$), this variant is called the Linear Tower of Hanoi. For more information, see chapter $8$ in \cite{hinz2018tower}.\\ 
It is proven that the only solvable variants based on disc movement restrictions are those  having a strongly connected movement digraph, which consists of three vertices that represent the three pegs, and an arc from vertex $u$ to vertex $v$ exists if and only if moves from peg $u$ to peg $v$ are allowed,  Figure \ref{GFC} shows the only five strongly connected digraphs on three vertices. These are the only solvable variants that are based on restricting disc movements on three pegs.
	\begin{figure}[H]
	    \centering
\tikzset{every picture/.style={line width=0.75pt}} 

\begin{tikzpicture}[x=0.75pt,y=0.75pt,yscale=-1,xscale=1]

\draw  [fill={rgb, 255:red, 0; green, 0; blue, 0 }  ,fill opacity=1 ] (128.01,35.12) .. controls (128.01,31.86) and (130.5,29.21) .. (133.56,29.21) .. controls (136.63,29.21) and (139.11,31.86) .. (139.11,35.12) .. controls (139.11,38.39) and (136.63,41.03) .. (133.56,41.03) .. controls (130.5,41.03) and (128.01,38.39) .. (128.01,35.12) -- cycle ;
\draw  [fill={rgb, 255:red, 0; green, 0; blue, 0 }  ,fill opacity=1 ] (191.63,146.25) .. controls (191.63,142.99) and (194.11,140.34) .. (197.18,140.34) .. controls (200.24,140.34) and (202.73,142.99) .. (202.73,146.25) .. controls (202.73,149.51) and (200.24,152.16) .. (197.18,152.16) .. controls (194.11,152.16) and (191.63,149.51) .. (191.63,146.25) -- cycle ;
\draw  [fill={rgb, 255:red, 0; green, 0; blue, 0 }  ,fill opacity=1 ] (65.14,146.25) .. controls (65.14,142.99) and (67.62,140.34) .. (70.68,140.34) .. controls (73.75,140.34) and (76.23,142.99) .. (76.23,146.25) .. controls (76.23,149.51) and (73.75,152.16) .. (70.68,152.16) .. controls (67.62,152.16) and (65.14,149.51) .. (65.14,146.25) -- cycle ;
\draw  [fill={rgb, 255:red, 0; green, 0; blue, 0 }  ,fill opacity=1 ] (328.01,35.12) .. controls (328.01,31.86) and (330.5,29.21) .. (333.56,29.21) .. controls (336.63,29.21) and (339.11,31.86) .. (339.11,35.12) .. controls (339.11,38.39) and (336.63,41.03) .. (333.56,41.03) .. controls (330.5,41.03) and (328.01,38.39) .. (328.01,35.12) -- cycle ;
\draw  [fill={rgb, 255:red, 0; green, 0; blue, 0 }  ,fill opacity=1 ] (391.63,146.25) .. controls (391.63,142.99) and (394.11,140.34) .. (397.18,140.34) .. controls (400.24,140.34) and (402.73,142.99) .. (402.73,146.25) .. controls (402.73,149.51) and (400.24,152.16) .. (397.18,152.16) .. controls (394.11,152.16) and (391.63,149.51) .. (391.63,146.25) -- cycle ;
\draw  [fill={rgb, 255:red, 0; green, 0; blue, 0 }  ,fill opacity=1 ] (265.14,146.25) .. controls (265.14,142.99) and (267.62,140.34) .. (270.68,140.34) .. controls (273.75,140.34) and (276.23,142.99) .. (276.23,146.25) .. controls (276.23,149.51) and (273.75,152.16) .. (270.68,152.16) .. controls (267.62,152.16) and (265.14,149.51) .. (265.14,146.25) -- cycle ;
\draw  [fill={rgb, 255:red, 0; green, 0; blue, 0 }  ,fill opacity=1 ] (528.01,35.12) .. controls (528.01,31.86) and (530.5,29.21) .. (533.56,29.21) .. controls (536.63,29.21) and (539.11,31.86) .. (539.11,35.12) .. controls (539.11,38.39) and (536.63,41.03) .. (533.56,41.03) .. controls (530.5,41.03) and (528.01,38.39) .. (528.01,35.12) -- cycle ;
\draw  [fill={rgb, 255:red, 0; green, 0; blue, 0 }  ,fill opacity=1 ] (591.63,146.25) .. controls (591.63,142.99) and (594.11,140.34) .. (597.18,140.34) .. controls (600.24,140.34) and (602.73,142.99) .. (602.73,146.25) .. controls (602.73,149.51) and (600.24,152.16) .. (597.18,152.16) .. controls (594.11,152.16) and (591.63,149.51) .. (591.63,146.25) -- cycle ;
\draw  [fill={rgb, 255:red, 0; green, 0; blue, 0 }  ,fill opacity=1 ] (465.14,146.25) .. controls (465.14,142.99) and (467.62,140.34) .. (470.68,140.34) .. controls (473.75,140.34) and (476.23,142.99) .. (476.23,146.25) .. controls (476.23,149.51) and (473.75,152.16) .. (470.68,152.16) .. controls (467.62,152.16) and (465.14,149.51) .. (465.14,146.25) -- cycle ;
\draw  [fill={rgb, 255:red, 0; green, 0; blue, 0 }  ,fill opacity=1 ] (228.01,233.12) .. controls (228.01,229.86) and (230.5,227.21) .. (233.56,227.21) .. controls (236.63,227.21) and (239.11,229.86) .. (239.11,233.12) .. controls (239.11,236.39) and (236.63,239.03) .. (233.56,239.03) .. controls (230.5,239.03) and (228.01,236.39) .. (228.01,233.12) -- cycle ;
\draw  [fill={rgb, 255:red, 0; green, 0; blue, 0 }  ,fill opacity=1 ] (291.63,344.25) .. controls (291.63,340.99) and (294.11,338.34) .. (297.18,338.34) .. controls (300.24,338.34) and (302.73,340.99) .. (302.73,344.25) .. controls (302.73,347.51) and (300.24,350.16) .. (297.18,350.16) .. controls (294.11,350.16) and (291.63,347.51) .. (291.63,344.25) -- cycle ;
\draw  [fill={rgb, 255:red, 0; green, 0; blue, 0 }  ,fill opacity=1 ] (165.14,344.25) .. controls (165.14,340.99) and (167.62,338.34) .. (170.68,338.34) .. controls (173.75,338.34) and (176.23,340.99) .. (176.23,344.25) .. controls (176.23,347.51) and (173.75,350.16) .. (170.68,350.16) .. controls (167.62,350.16) and (165.14,347.51) .. (165.14,344.25) -- cycle ;
\draw  [fill={rgb, 255:red, 0; green, 0; blue, 0 }  ,fill opacity=1 ] (428.01,233.12) .. controls (428.01,229.86) and (430.5,227.21) .. (433.56,227.21) .. controls (436.63,227.21) and (439.11,229.86) .. (439.11,233.12) .. controls (439.11,236.39) and (436.63,239.03) .. (433.56,239.03) .. controls (430.5,239.03) and (428.01,236.39) .. (428.01,233.12) -- cycle ;
\draw  [fill={rgb, 255:red, 0; green, 0; blue, 0 }  ,fill opacity=1 ] (491.63,344.25) .. controls (491.63,340.99) and (494.11,338.34) .. (497.18,338.34) .. controls (500.24,338.34) and (502.73,340.99) .. (502.73,344.25) .. controls (502.73,347.51) and (500.24,350.16) .. (497.18,350.16) .. controls (494.11,350.16) and (491.63,347.51) .. (491.63,344.25) -- cycle ;
\draw  [fill={rgb, 255:red, 0; green, 0; blue, 0 }  ,fill opacity=1 ] (365.14,344.25) .. controls (365.14,340.99) and (367.62,338.34) .. (370.68,338.34) .. controls (373.75,338.34) and (376.23,340.99) .. (376.23,344.25) .. controls (376.23,347.51) and (373.75,350.16) .. (370.68,350.16) .. controls (367.62,350.16) and (365.14,347.51) .. (365.14,344.25) -- cycle ;

\draw  [color={rgb, 255:red, 0; green, 0; blue, 0 }  ,draw opacity=0 ]  (133.19, 37.09) circle [x radius= 13.6, y radius= 13.6]   ;
\draw (127.19,29.49) node [anchor=north west][inner sep=0.75pt]    {};
\draw  [color={rgb, 255:red, 0; green, 0; blue, 0 }  ,draw opacity=0 ]  (70.32, 147.43) circle [x radius= 13.6, y radius= 13.6]   ;
\draw (64.32,139.83) node [anchor=north west][inner sep=0.75pt]    {};
\draw  [color={rgb, 255:red, 0; green, 0; blue, 0 }  ,draw opacity=0 ]  (196.07, 147.43) circle [x radius= 13.6, y radius= 13.6]   ;
\draw (190.07,139.83) node [anchor=north west][inner sep=0.75pt]    {};
\draw (139.77,22.4) node [anchor=north west][inner sep=0.75pt]    {$1$};
\draw (202.64,146.14) node [anchor=north west][inner sep=0.75pt]    {$2$};
\draw (51,146.14) node [anchor=north west][inner sep=0.75pt]    {$3$};
\draw (124.82,94.52) node [anchor=north west][inner sep=0.75pt]    {$\stackrel {\leftrightarrow}{K}_{3}$};
\draw  [color={rgb, 255:red, 0; green, 0; blue, 0 }  ,draw opacity=0 ]  (333.19, 37.09) circle [x radius= 13.6, y radius= 13.6]   ;
\draw (327.19,29.49) node [anchor=north west][inner sep=0.75pt]    {};
\draw  [color={rgb, 255:red, 0; green, 0; blue, 0 }  ,draw opacity=0 ]  (270.32, 147.43) circle [x radius= 13.6, y radius= 13.6]   ;
\draw (264.32,139.83) node [anchor=north west][inner sep=0.75pt]    {};
\draw  [color={rgb, 255:red, 0; green, 0; blue, 0 }  ,draw opacity=0 ]  (396.07, 147.43) circle [x radius= 13.6, y radius= 13.6]   ;
\draw (390.07,139.83) node [anchor=north west][inner sep=0.75pt]    {};
\draw (339.77,22.4) node [anchor=north west][inner sep=0.75pt]    {$1$};
\draw (402.64,146.14) node [anchor=north west][inner sep=0.75pt]    {$2$};
\draw (251,146.14) node [anchor=north west][inner sep=0.75pt]    {$3$};
\draw (324.82,94.52) node [anchor=north west][inner sep=0.75pt]    {$\stackrel {\leftrightarrow}{L}_{3}$};
\draw  [color={rgb, 255:red, 0; green, 0; blue, 0 }  ,draw opacity=0 ]  (533.19, 37.09) circle [x radius= 13.6, y radius= 13.6]   ;
\draw (527.19,29.49) node [anchor=north west][inner sep=0.75pt]    {};
\draw  [color={rgb, 255:red, 0; green, 0; blue, 0 }  ,draw opacity=0 ]  (470.32, 147.43) circle [x radius= 13.6, y radius= 13.6]   ;
\draw (464.32,139.83) node [anchor=north west][inner sep=0.75pt]    {};
\draw  [color={rgb, 255:red, 0; green, 0; blue, 0 }  ,draw opacity=0 ]  (596.07, 147.43) circle [x radius= 13.6, y radius= 13.6]   ;
\draw (590.07,139.83) node [anchor=north west][inner sep=0.75pt]    {};
\draw (539.77,22.4) node [anchor=north west][inner sep=0.75pt]    {$1$};
\draw (602.64,146.14) node [anchor=north west][inner sep=0.75pt]    {$2$};
\draw (451,146.14) node [anchor=north west][inner sep=0.75pt]    {$3$};
\draw (411,297.4) node [anchor=north west][inner sep=0.75pt]    {$\overrightarrow{C_{3+}}$};
\draw  [color={rgb, 255:red, 0; green, 0; blue, 0 }  ,draw opacity=0 ]  (233.19, 235.09) circle [x radius= 13.6, y radius= 13.6]   ;
\draw (227.19,227.49) node [anchor=north west][inner sep=0.75pt]    {};
\draw  [color={rgb, 255:red, 0; green, 0; blue, 0 }  ,draw opacity=0 ]  (170.32, 345.43) circle [x radius= 13.6, y radius= 13.6]   ;
\draw (164.32,337.83) node [anchor=north west][inner sep=0.75pt]    {};
\draw  [color={rgb, 255:red, 0; green, 0; blue, 0 }  ,draw opacity=0 ]  (296.07, 345.43) circle [x radius= 13.6, y radius= 13.6]   ;
\draw (290.07,337.83) node [anchor=north west][inner sep=0.75pt]    {};
\draw (239.77,220.4) node [anchor=north west][inner sep=0.75pt]    {$1$};
\draw (302.64,344.14) node [anchor=north west][inner sep=0.75pt]    {$2$};
\draw (151,344.14) node [anchor=north west][inner sep=0.75pt]    {$3$};
\draw (224.82,292.52) node [anchor=north west][inner sep=0.75pt]    {$\stackrel {\leftrightarrow}{K}_{3-}$};
\draw  [color={rgb, 255:red, 0; green, 0; blue, 0 }  ,draw opacity=0 ]  (433.19, 235.09) circle [x radius= 13.6, y radius= 13.6]   ;
\draw (427.19,227.49) node [anchor=north west][inner sep=0.75pt]    {};
\draw  [color={rgb, 255:red, 0; green, 0; blue, 0 }  ,draw opacity=0 ]  (370.32, 345.43) circle [x radius= 13.6, y radius= 13.6]   ;
\draw (364.32,337.83) node [anchor=north west][inner sep=0.75pt]    {};
\draw  [color={rgb, 255:red, 0; green, 0; blue, 0 }  ,draw opacity=0 ]  (496.07, 345.43) circle [x radius= 13.6, y radius= 13.6]   ;
\draw (490.07,337.83) node [anchor=north west][inner sep=0.75pt]    {};
\draw (439.77,220.4) node [anchor=north west][inner sep=0.75pt]    {$1$};
\draw (502.64,344.14) node [anchor=north west][inner sep=0.75pt]    {$2$};
\draw (351,344.14) node [anchor=north west][inner sep=0.75pt]    {$3$};
\draw (522,94.4) node [anchor=north west][inner sep=0.75pt]    {$\overrightarrow{C_{3}}$};
\draw    (70.04,133.83) .. controls (81.16,101.4) and (97.7,72.06) .. (119.65,45.8) ;
\draw [shift={(121.35,43.79)}, rotate = 130.42] [fill={rgb, 255:red, 0; green, 0; blue, 0 }  ][line width=0.08]  [draw opacity=0] (8.93,-4.29) -- (0,0) -- (8.93,4.29) -- cycle    ;
\draw    (82.14,140.7) .. controls (115.32,135.23) and (148.51,135.09) .. (181.69,140.29) ;
\draw [shift={(184.24,140.7)}, rotate = 189.37] [fill={rgb, 255:red, 0; green, 0; blue, 0 }  ][line width=0.08]  [draw opacity=0] (8.93,-4.29) -- (0,0) -- (8.93,4.29) -- cycle    ;
\draw    (133.47,50.69) .. controls (122.35,83.13) and (105.81,112.47) .. (83.86,138.72) ;
\draw [shift={(82.16,140.74)}, rotate = 310.43] [fill={rgb, 255:red, 0; green, 0; blue, 0 }  ][line width=0.08]  [draw opacity=0] (8.93,-4.29) -- (0,0) -- (8.93,4.29) -- cycle    ;
\draw    (145.03,43.79) .. controls (167.27,69.89) and (184.08,99.07) .. (195.48,131.34) ;
\draw [shift={(196.35,133.83)}, rotate = 251.07] [fill={rgb, 255:red, 0; green, 0; blue, 0 }  ][line width=0.08]  [draw opacity=0] (8.93,-4.29) -- (0,0) -- (8.93,4.29) -- cycle    ;
\draw    (184.23,140.74) .. controls (161.99,114.64) and (145.18,85.46) .. (133.78,53.18) ;
\draw [shift={(132.91,50.69)}, rotate = 71.08] [fill={rgb, 255:red, 0; green, 0; blue, 0 }  ][line width=0.08]  [draw opacity=0] (8.93,-4.29) -- (0,0) -- (8.93,4.29) -- cycle    ;
\draw    (184.25,154.16) .. controls (151.06,159.63) and (117.87,159.77) .. (84.69,154.57) ;
\draw [shift={(82.14,154.16)}, rotate = 9.36] [fill={rgb, 255:red, 0; green, 0; blue, 0 }  ][line width=0.08]  [draw opacity=0] (8.93,-4.29) -- (0,0) -- (8.93,4.29) -- cycle    ;
\draw    (282.14,140.7) .. controls (315.32,135.23) and (348.51,135.09) .. (381.69,140.29) ;
\draw [shift={(384.24,140.7)}, rotate = 189.37] [fill={rgb, 255:red, 0; green, 0; blue, 0 }  ][line width=0.08]  [draw opacity=0] (8.93,-4.29) -- (0,0) -- (8.93,4.29) -- cycle    ;
\draw    (345.03,43.79) .. controls (367.27,69.89) and (384.08,99.07) .. (395.48,131.34) ;
\draw [shift={(396.35,133.83)}, rotate = 251.07] [fill={rgb, 255:red, 0; green, 0; blue, 0 }  ][line width=0.08]  [draw opacity=0] (8.93,-4.29) -- (0,0) -- (8.93,4.29) -- cycle    ;
\draw    (384.23,140.74) .. controls (361.99,114.64) and (345.18,85.46) .. (333.78,53.18) ;
\draw [shift={(332.91,50.69)}, rotate = 71.08] [fill={rgb, 255:red, 0; green, 0; blue, 0 }  ][line width=0.08]  [draw opacity=0] (8.93,-4.29) -- (0,0) -- (8.93,4.29) -- cycle    ;
\draw    (384.25,154.16) .. controls (351.06,159.63) and (317.87,159.77) .. (284.69,154.57) ;
\draw [shift={(282.14,154.16)}, rotate = 9.36] [fill={rgb, 255:red, 0; green, 0; blue, 0 }  ][line width=0.08]  [draw opacity=0] (8.93,-4.29) -- (0,0) -- (8.93,4.29) -- cycle    ;
\draw    (475.02,134.67) .. controls (485.73,103.3) and (501.71,74.94) .. (522.96,49.59) ;
\draw [shift={(524.61,47.64)}, rotate = 130.51] [fill={rgb, 255:red, 0; green, 0; blue, 0 }  ][line width=0.08]  [draw opacity=0] (8.93,-4.29) -- (0,0) -- (8.93,4.29) -- cycle    ;
\draw    (541.78,47.64) .. controls (563.31,72.84) and (579.56,101.05) .. (590.53,132.26) ;
\draw [shift={(591.37,134.67)}, rotate = 251.16] [fill={rgb, 255:red, 0; green, 0; blue, 0 }  ][line width=0.08]  [draw opacity=0] (8.93,-4.29) -- (0,0) -- (8.93,4.29) -- cycle    ;
\draw    (582.6,149.37) .. controls (550.48,154.7) and (518.37,154.84) .. (486.25,149.77) ;
\draw [shift={(483.78,149.37)}, rotate = 9.42] [fill={rgb, 255:red, 0; green, 0; blue, 0 }  ][line width=0.08]  [draw opacity=0] (8.93,-4.29) -- (0,0) -- (8.93,4.29) -- cycle    ;
\draw    (175.02,332.67) .. controls (185.73,301.3) and (201.71,272.94) .. (222.96,247.59) ;
\draw [shift={(224.61,245.64)}, rotate = 130.51] [fill={rgb, 255:red, 0; green, 0; blue, 0 }  ][line width=0.08]  [draw opacity=0] (8.93,-4.29) -- (0,0) -- (8.93,4.29) -- cycle    ;
\draw    (182.14,338.7) .. controls (215.32,333.23) and (248.51,333.09) .. (281.69,338.29) ;
\draw [shift={(284.24,338.7)}, rotate = 189.37] [fill={rgb, 255:red, 0; green, 0; blue, 0 }  ][line width=0.08]  [draw opacity=0] (8.93,-4.29) -- (0,0) -- (8.93,4.29) -- cycle    ;
\draw    (245.03,241.79) .. controls (267.27,267.89) and (284.08,297.07) .. (295.48,329.34) ;
\draw [shift={(296.35,331.83)}, rotate = 251.07] [fill={rgb, 255:red, 0; green, 0; blue, 0 }  ][line width=0.08]  [draw opacity=0] (8.93,-4.29) -- (0,0) -- (8.93,4.29) -- cycle    ;
\draw    (284.23,338.74) .. controls (261.99,312.64) and (245.18,283.46) .. (233.78,251.18) ;
\draw [shift={(232.91,248.69)}, rotate = 71.08] [fill={rgb, 255:red, 0; green, 0; blue, 0 }  ][line width=0.08]  [draw opacity=0] (8.93,-4.29) -- (0,0) -- (8.93,4.29) -- cycle    ;
\draw    (284.25,352.16) .. controls (251.06,357.63) and (217.87,357.77) .. (184.69,352.57) ;
\draw [shift={(182.14,352.16)}, rotate = 9.36] [fill={rgb, 255:red, 0; green, 0; blue, 0 }  ][line width=0.08]  [draw opacity=0] (8.93,-4.29) -- (0,0) -- (8.93,4.29) -- cycle    ;
\draw    (375.02,332.67) .. controls (385.73,301.3) and (401.71,272.94) .. (422.96,247.59) ;
\draw [shift={(424.61,245.64)}, rotate = 130.51] [fill={rgb, 255:red, 0; green, 0; blue, 0 }  ][line width=0.08]  [draw opacity=0] (8.93,-4.29) -- (0,0) -- (8.93,4.29) -- cycle    ;
\draw    (445.03,241.79) .. controls (467.27,267.89) and (484.08,297.07) .. (495.48,329.34) ;
\draw [shift={(496.35,331.83)}, rotate = 251.07] [fill={rgb, 255:red, 0; green, 0; blue, 0 }  ][line width=0.08]  [draw opacity=0] (8.93,-4.29) -- (0,0) -- (8.93,4.29) -- cycle    ;
\draw    (484.23,338.74) .. controls (461.99,312.64) and (445.18,283.46) .. (433.78,251.18) ;
\draw [shift={(432.91,248.69)}, rotate = 71.08] [fill={rgb, 255:red, 0; green, 0; blue, 0 }  ][line width=0.08]  [draw opacity=0] (8.93,-4.29) -- (0,0) -- (8.93,4.29) -- cycle    ;
\draw    (482.6,347.37) .. controls (450.48,352.7) and (418.37,352.84) .. (386.25,347.77) ;
\draw [shift={(383.78,347.37)}, rotate = 9.42] [fill={rgb, 255:red, 0; green, 0; blue, 0 }  ][line width=0.08]  [draw opacity=0] (8.93,-4.29) -- (0,0) -- (8.93,4.29) -- cycle    ;

\end{tikzpicture}

	    \caption{Strongly connected digraphs on three vertices}
	    \label{GFC}
	\end{figure}
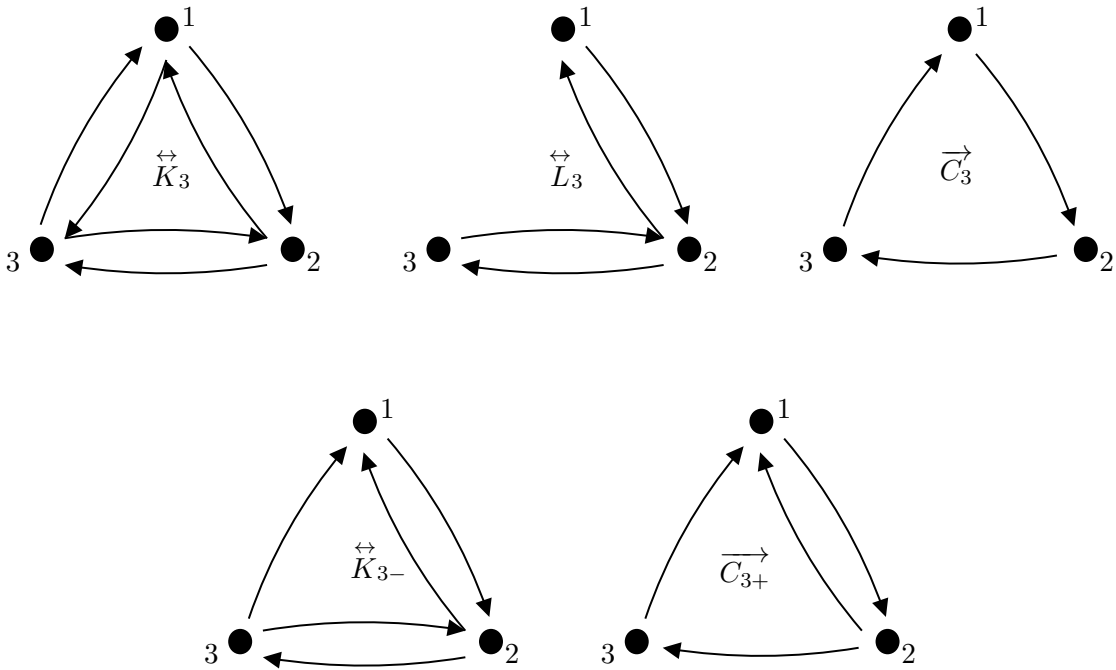
Sapir in \cite{sapir2004tower}, proposed an algorithm that solves any such variant in a minimum number of moves. By carefully choosing the weights on arcs, we can design a  class of WTH instances that corresponds to one of the five variants in such a way  that these instances have an optimal solution that respects the restriction of disc moves of the corresponding variant; Which implies that algorithm \textsc{WTHD} can be used to solve any variant with restricted disc moves.\\

Let $X\in \{\overleftrightarrow{K}_{3},\overleftrightarrow{L}_{3},\overrightarrow{C} _{3},\overleftrightarrow{K}_{3-},\overrightarrow{C}_{3+}\}$ be a variant of the Tower of Hanoi with restricted disc moves and $A(X)$ its digraph movement. To forbid  movements from peg $i$ to  peg $j$, it suffices to choose $w_{ij}$ large enough so that the algorithm \textsc{WTHD} never chooses to move a disc from  peg $i$ to peg $j$. As the expression \textit{"large enough"} is very relative, we expand our study to search about the minimum value that can be given to $w_{ij}$ so that there is at least one optimal solution without any disc movement from peg $i$ to peg $j$. As a first observation, it is clear that it is always sufficient to choose $w_{ij}=\infty$ so that the optimal solution(s) does not admit movements from peg $i$ to peg $j$ simply because $\infty$ is always large enough.\\

The following lemma characterize the condition that must be fulfilled so that an instance of the WTH problem has an optimal solution with no disc moves  from certain peg to another.

\begin{lemma}\label{lemma1}
A WTH instance of $n$ discs, has no optimal solution with movements of discs from peg $i$ to peg $j$ if and only if 

\begin{align*}
w_{ij}&>w_{ik}+w_{kj}+\max\{0,2C_{n-1}^{i,j}+C_{n-1}^{j,i}-C_{n -1}^{i,k}-C_{n-1}^{k,j} \},\; and\\
C_{n}^{i,j}&=2C_{n-1}^{i,j}+C_{n-1}^{j,i}+w_{ik}+w_{kj}.
\end{align*}
\end{lemma}

\begin{proof}
If no optimal solution has a disc movement from peg $i$ to peg $j$ then the solution provided by Algorithm \textsc{WTHD} has also no disc movement from peg $i$ to peg $j$ which implies  $w_{ij}>w_{ik}+w_{kj}$ and 
$$C_{n-1}^{i,k}+C_{n-1}^{k,j}+w_{ij}>2C_{n-1}^{i,j}+C_{n-1 }^{j,i}+w_{ik}+w_{kj}.$$
Then $C_{n}^{i,j}=2C_{n-1}^{i,j}+C_{n-1}^{j,i}+w_{ik}+w_{kj}$ and 
\begin{equation*}
    w_{ij}>\max\{w_{ik}+w_{kj},2C_{n-1}^{i,j}+C_{n-1}^{j,i}+w_{ik}+ w_{kj}-C_{n-1}^{i,k}-C_{n-1}^{k,j}\}.
\end{equation*}
Thus,
\begin{equation*}
    w_{ij}>w_{ik}+w_{kj}+\max\{0,2C_{n-1}^{i,j}+C_{n-1}^{j,i}-C_{n -1}^{i,k}-C_{n-1}^{k,j} \}.
\end{equation*}
In the other side if 
\begin{align*}
w_{ij}&>w_{ik}+w_{kj}+\max\{0,2C_{n-1}^{i,j}+C_{n-1}^{j,i}-C_{n -1}^{i,k}-C_{n-1}^{k,j} \},\; and\\
C_{n}^{i,j}&=2C_{n-1}^{i,j}+C_{n-1}^{j,i}+w_{ik}+w_{kj}.
\end{align*}
We obtain
\begin{align*}
w_{ij}&>w_{ik}+w_{kj},\; and\\
C_{n -1}^{i,k}+C_{n-1}^{k,j}+w_{ij}&<2C_{n-1}^{i,j}+C_{n-1}^{j,i}+w_{ik}+w_{kj}.
\end{align*}
Therefor, Algorithm \textsc{WTHD} will never choose  to move a disc from peg $i$ to peg $j$. Hence no optimal solution will have a disc move from peg $i$ to peg $j$. 
\end{proof}
The next theorem is a generalization of Lemma \ref{lemma1}, it characterizes a class of WTH instances that correspond to one of the five variants with restricted disc moves, i.e. for each variant $X\in \{\overleftrightarrow{K}_{3},\overleftrightarrow{L}_{3},\overrightarrow{C} _{3},\overleftrightarrow{K}_{3-},\overrightarrow{C}_{3+}\}$ it gives the characteristics of the weights $(w_{ij})_{1\leq i,j\leq 3 }$ so that at least one optimal solution is of the same type as the solution of variant $X$.
\begin{theorem}\label{theorem2}
Let $X$ be a variant of the Tower of Hanoi with restricted disc moves, and $D=(V(D),A(D))\in\{\overleftrightarrow{K}_{3},\overleftrightarrow {L}_{3},\overrightarrow{C}_{3},\overleftrightarrow{K}_{3-},\overrightarrow{C}_{3+}\}$ its movement digraph. An instance of the WTH has an optimal solution corresponds to the restriction on disc moves of variant $X$  if and only if for all $(i,j)\notin A(D)$, $i\neq j$, we have
\begin{equation*}
w_{ij}>w_{ik}+w_{kj}+\max\{0,2C_{n-1}^{i,j}+C_{n-1}^{j,i}-C_{n -1}^{i,k}-C_{n-1}^{k,j} \},
\end{equation*}
and $C_{n}^{i,j}=2C_{n-1}^{i,j}+C_{n-1}^{j,i}+w_{ik}+w_{kj}$.
\end{theorem}
\begin{proof}
for each $(i,j)\notin A(D)$, $i\neq j$, we do the same reasoning as the previous proof.
\end{proof}
We present bellow two direct results of Theorem \ref{theorem2} which are  Corollaries 
\ref{lin} and \ref{cyc} that characterize the class of WTH instances in which the optimal solution respects the restrictions on disc moves of the linear variant and the cyclic variant respectively.
\begin{corollary}\label{lin}

An instance of the WTH, accepts the linear variant solution as an optimal solution if and only if
\begin{equation*}
w_{13}>w_{12}+w_{23}+\max\{0,2C_{n-1}^{1.3}+C_{n-1}^{3.1}-C_{n -1}^{1,2}-C_{n-1}^{2,3} \}
\end{equation*}
and
\begin{equation*}
w_{31}>w_{32}+w_{21}+\max\{0,2C_{n-1}^{3.1}+C_{n-1}^{1.3}-C_{n -1}^{3,2}-C_{n-1}^{2,1} \}.
\end{equation*}
\end{corollary}
\begin{corollary}\label{cyc}

An instance of the WTH, accepts the cyclic variant solution as an optimal solution if and only if
\begin{equation*}
w_{13}>w_{12}+w_{23}+\max\{0,2C_{n-1}^{1.3}+C_{n-1}^{3.1}-C_{n -1}^{1,2}-C_{n-1}^{2,3} \},
\end{equation*}
\begin{equation*}
w_{32}>w_{31}+w_{12}+\max\{0,2C_{n-1}^{3,2}+C_{n-1}^{2,3}-C_{n -1}^{3,1}-C_{n-1}^{1,3} \},
\end{equation*}
\begin{equation*}
w_{21}>w_{23}+w_{31}+\max\{0,2C_{n-1}^{2,1}+C_{n-1}^{1,2}-C_{n -1}^{2,3}-C_{n-1}^{3,1} \}.
\end{equation*}
\end{corollary}

\section{Conclusion}
A new generalization of the Tower of Hanoi puzzle was presented in this paper, in this generalization the moves between pegs are weighted and the objective is to solve the puzzle while minimizing the total cost. The special thing about this new generalization is that it generalize the concept of the minimum number of moves in the classical Tower of Hanoi so that each move is measured using a positive weight. We gave a recursive formula of the total cost in terms of the number of discs, also we introduced an optimal algorithm to solve this new optimization problem using dynamic programming. We established also  the relationship between the weighted Tower of Hanoi and its variants with restricted disc moves. Furthermore, due to the complexity time of the recursive algorithm proposed in this paper, one could think about using heuristics and meta-heuristics to solve the WTH problem.



\bibliographystyle{plain}
\bibliography{bibou.bib}

\end{document}